\documentclass[11pt]{article}
\usepackage{fullpage}

\usepackage{latexsym}
\usepackage{amsmath}
\usepackage{amssymb}
\usepackage{amsthm}
\usepackage{hyperref}
\usepackage{cite}
\usepackage{graphicx}
\usepackage{color}
\usepackage{enumitem}
\setlist[description]{font=\normalfont\itshape\textbullet\space}

\newtheorem{theorem}{Theorem}
\newtheorem{theorem*}{Theorem}
\newtheorem{corollary}[theorem]{Corollary}
\newtheorem{lemma}[theorem]{Lemma}

\newtheorem{proposition}[theorem]{Proposition}
\newtheorem{definition}[theorem]{Definition}
\newtheorem{claim}[theorem]{Claim}
\newtheorem{fact}[theorem]{Fact}

\theoremstyle{definition}
\newtheorem{problem}[theorem]{Problem}
\newtheorem{example}[theorem]{Example}
\newtheorem{remark}[theorem]{Remark}

\renewcommand{\vec}[1]{\mathbf{#1}}
\newcommand{\poly}{\mathrm{poly}}

\newcommand{\GL}{\mathrm{GL}}
\newcommand{\E}{\mathbb{E}}
\newcommand{\F}{\mathbb{F}}
\newcommand{\K}{\mathbb{K}}
\newcommand{\Z}{\mathbb{Z}}
\newcommand{\Q}{\mathbb{Q}}
\newcommand{\R}{\mathbb{R}}
\newcommand{\C}{\mathbb{C}}
\newcommand{\qH}{\mathbb{H}}
\newcommand{\N}{\mathbb{N}}

\renewcommand{\hom}{\mathrm{Hom}}
\newcommand{\Ann}{\mathrm{Ann}}
\newcommand{\diag}{\mathrm{diag}}

\newcommand{\rk}{\mathrm{rk}}

\newcommand{\im}{\mathrm{im}}

\newcommand{\cB}{\mathcal{B}}
\newcommand{\cC}{\mathcal{C}}

\newcommand{\Adj}{\mathrm{Adj}}
\newcommand{\Rad}{\mathrm{Rad}}

\renewcommand{\iff}{\Leftrightarrow}

\newcommand{\fdchar}{\mathrm{char}}

\newcommand{\cA}{\mathfrak{A}}

\newcommand{\vecB}{\vec{B}}
\newcommand{\vecC}{\vec{C}}
\newcommand{\vzero}{\mathbf{0}}

\newcommand{\vecf}{\vec{f}}
\newcommand{\vecg}{\vec{g}}

\def\A{{\mathfrak A}}
\def\D{{\mathcal D}}
\def\Calg{{\mathcal C}}
\def\B{{\mathfrak B}}

\title{Algorithms based on $*$-algebras, 
and their applications 
to isomorphism of polynomials with one secret, group isomorphism, and polynomial 
identity testing\footnote{A preliminary version of this paper appeared in SODA 
2018 as \cite{IQ18}.}
}

\author{
G\'abor Ivanyos\thanks{Institute for Computer Science and Control, Hungarian 
Academy of Sciences, 
Budapest, Hungary 
({\tt Gabor.Ivanyos@sztaki.mta.hu}).} 
\and
 Youming Qiao\thanks{Centre for Quantum Software and Information, 
 University of Technology Sydney, Australia ({\tt Youming.Qiao@uts.edu.au}) 
 }
}

\date{\today}

\begin{document}

\maketitle

\begin{abstract} 
We consider two basic algorithmic 
problems concerning tuples 
of (skew-)symmetric matrices. 
The first problem asks to decide, given two tuples of (skew-)symmetric matrices 
$(B_1, \dots, B_m)$ and $(C_1, \dots, C_m)$, whether 
there exists an invertible matrix $A$ such that for every $i\in\{1, \dots, m\}$, 
$A^tB_iA=C_i$. We show that this problem can be solved in randomized polynomial 
time over finite fields of odd size, the reals, and the 
complex numbers. 
The second problem asks to decide, given a tuple of square matrices $(B_1, 
\dots, B_m)$, whether there exist invertible matrices $A$ and 
$D$, such that for every $i\in\{1, \dots, m\}$, $AB_iD$ is (skew-)symmetric.
We show that this problem can be solved in deterministic polynomial time over 
fields of characteristic not $2$. For 
both problems we exploit the structure 
of the underlying $*$-algebras 
(algebras with an involutive anti-automorphism), and 
utilize results and methods from the module isomorphism problem.

Applications of our results range from multivariate 
cryptography, group 
isomorphism, to polynomial identity testing. Specifically, these results imply 
efficient 
algorithms for the following 
problems. (1) Test 
isomorphism of quadratic forms with one secret over a finite field of odd 
size. This problem belongs to a family of problems that serves as the security 
basis of certain 
authentication schemes proposed  
by Patarin (Eurocrypt 1996). 
(2) Test 
isomorphism of $p$-groups of class 2 and exponent $p$ ($p$ odd) with order 
$p^\ell$ 
in time polynomial in the group order, when the commutator 
subgroup is of order $p^{O(\sqrt{\ell})}$. 
(3) Deterministically reveal two families of singularity witnesses caused by the 
skew-symmetric 
structure. This represents a natural next step for the polynomial  
identity testing 
problem, 
in the 
direction set up by the recent resolution of the non-commutative rank problem 
(Garg-Gurvits-Oliveira-Wigderson, 
FOCS 2016; Ivanyos-Qiao-Subrahmanyam, ITCS 2017).
\end{abstract}

\section{Introduction}

We consider two basic algorithmic problems concerning tuples of (skew-)symmetric 
matrices. For convenience, for $\epsilon\in \{1, -1\}$, we say an $n\times n$ 
matrix $B$ is $\epsilon$-symmetric, if $B^t=\epsilon B$. Clearly, when 
$\epsilon=1$ (resp. $\epsilon=-1$), $B$ is symmetric (resp. skew-symmetric).

The first problem asks to decide, given two tuples of $n\times n$
$\epsilon$-symmetric matrices $(B_1, \dots, B_m)$ and $(C_1, \dots, C_m)$, whether 
there exists an invertible $n\times n$ matrix $A$, such that $\forall i\in[m]$, 
$A^tB_iA=C_i$. 
We call this problem \emph{the isometry problem for $\epsilon$-symmetric matrix 
tuples}. 
We show that this problem can be solved in randomized polynomial time when the 
underlying field is a finite field of odd size, the field of real numbers, or 
the field of complex numbers. 

The second problem asks to decide, given a tuple of $n\times n$ matrices $(B_1, 
\dots, B_m)$, whether there exist invertible $n\times n$ matrices $A$ and $D$, 
such that 
$\forall i\in[m]$, $AB_iD$ is $\epsilon$-symmetric. 
We call this problem the \emph{$\epsilon$-symmetrization problem for matrix 
tuples}. 
We show that this problem can be solved in deterministic polynomial 
time, as long as the underlying field is not of characteristic $2$. 

At first sight, these two problems seem to be of interest mostly in computer 
algebra. 
However, as we explain below, these results are motivated by, and therefore have 
applications to, three seemingly 
unrelated research 
topics. These are multivariate cryptography, group isomorphism problem, and 
polynomial 
identity testing problem, which are traditionally studied in cryptography, 
computational 
group theory, and algebraic complexity theory, respectively. The algorithm for 
isometry testing of $\epsilon$-symmetric matrix tuples leads to substantial 
improvements over recent algorithms from multivariate cryptography and group 
isomorphism \cite{BFP15,BMW15}. In particular, the algorithm for isometry testing 
of symmetric matrix tuples completely settles the so-called 
Isomorphism of Quadratic Polynomials with One Secret problem over finite fields of 
odd 
size \cite{Pat96}. The 
algorithm for the $\epsilon$-symmetrization problem represents a
natural next step for the polynomial identity testing problem in the direction set 
up by the recent 
resolution of the non-commutative rank problem \cite{GGOW,IQS16a,IQS17}.

The algorithms for the isometry problem and the $\epsilon$-symmetrization problem 
share two key ingredients in common. The 
first one is to utilize the structure of $*$-algebras, that is algebras with an 
involutive anti-automorphism, underlying these problems. 
More specifically, given 
a field $\F$, an $\F$-algebra $\cA$ with 
anti-automorphism $*:\cA\to\cA$ of order at most 2 is termed as a $*$-algebra. We 
refer the reader to Section~\ref{sec:prel} for more details on the structure of 
$*$-algebras.
Our use of $*$-algebras 
is inspired by the works of J. B. Wilson, who 
pioneered the use of $*$-algebras in computing with $p$-groups 
\cite{Wil09,Wil09b,BW12}. The second one is 
the results and methods from the module isomorphism problem, which asks to decide, 
given two tuples of matrices $(B_1, \dots, B_m)$, $(C_1, \dots, C_m)$, whether 
there exists an invertible matrix $A$, such that $\forall i\in[m]$, $AB_i=C_iA$. 
This problem admits two deterministic efficient algorithms by \cite{CIK97,IKS10} 
and 
\cite{BL08}. These results and the techniques are used frequently in 
both algorithms. 

In this introduction, we first elaborate on the 
applications, from
Section~\ref{subsubsec:crypto} to~\ref{subsubsec:pit}. Since the 
applications span across three different 
areas, in order to provide the contexts for readers with different backgrounds, we 
shall not refrain from including 
certain 
background 
information, despite that it is well-known for researchers in the 
respective area. In Section~\ref{subsec:result}, we formally present the 
results, explain more on the two key ingredients shared by
both algorithms, and describe some open problems. 

We now set up some notation. $\F$, $\E$, and $\K$ are used to 
denote fields. $\F_q$ denotes the finite field of size $q$, 
$\R$ the real field, and $\C$ the complex field. Unless otherwise stated, we work 
with fields of characteristic not 
$2$. 
$M(n, \F)$ denotes the linear 
space 
of $n\times n$ matrices over $\F$, and $\GL(n, \F)$ the group of invertible 
matrices in $M(n, \F)$. $S^\epsilon(n, 
\F)$ denotes the linear space of $n\times n$ $\epsilon$-symmetric matrices over 
$\F$. 
We may write $M(n, q)$, $\GL(n, q)$, and $S^\epsilon(n, q)$ for $M(n, \F_q)$, 
$\GL(n, \F_q)$, and $S^\epsilon(n, \F_q)$, respectively. 
A \emph{matrix space} is a linear 
subspace of $M(n, \F)$, and $\langle\cdot\rangle$ denotes linear span. 
Let 
$\vecB=(B_1, \dots, B_m)\in M(n, \F)^m$ be a matrix tuple. 
For $A, D\in M(n, \F)$, $A\vecB D:=(AB_1D, 
\dots, AB_mD)$ and $\vecB^t:=(B_1^t, \dots, B_n^t)$.


\subsection{Multivariate cryptography}\label{subsubsec:crypto}

In 1996, Patarin proposed a family of 
asymmetric cryptography schemes based on equivalence of polynomials in 
\cite{Pat96}, which can be used for identification and signature schemes. 
One scheme in this family is based on 
the assumed hardness of the following problem. 
\begin{problem}{\sc (Isomorphism of Quadratic Forms with One Secret (IQF1S))}
Let $\vecf=(f_1, \dots, f_m)$ and $\vecg=(g_1, \dots, g_m)$ be two tuples of 
homogeneous 
quadratic polynomials in $n$ variables $\{x_1, \dots, x_n\}$ over a finite field 
$\F$. 
Decide if there exists $A\in \GL(n, \F)$ such that $\forall k\in[m]$, 
$f_k^A=g_k$, where $A=(a_{i,j})_{i,j\in[n]}$ acts on $\{x_1, \dots, 
x_n\}$ by sending $x_i$ to $\sum_{j\in[n]}a_{i,j}x_j$.
\end{problem}
For readers familiar with Patarin's work \cite{Pat96}, IQF1S is 
Patarin's Isomorphism of Polynomials with One Secret (IP1S) restricting to 
quadratic polynomials, which asks the same question but for possibly inhomogeneous 
quadratic polynomials and affine transformations.\footnote{Patarin's formulation 
is known to reduce to the 
formulation here \cite[Proposition 5]{BFP15}.} Such a restriction is well 
justified from the practical viewpoint, as it minimizes the public-key storage and 
improves the actual performance, so this has been studied most in the 
literature. Since Patarin's introduction of these problems, 
IQF1S and several related problems have been intensively studied 
\cite{PGC98,GMS03,Per05,FP06,Kay11,BFFP11,MPG13,BFV13,PFM14,BFP15}. 

Most notably, in \cite{BFP15}, Berthomieu et al.~presented an efficient
randomized algorithm for IQF1S
under the conditions that (1) $\vecf$ satisfies a regularity 
condition, namely that there exists a non-degenerate form in the linear span 
of $f_i$'s, (2) the underlying field is large enough and of characteristic not 
$2$, and (3) the desired solution 
may be from an extension field 
\cite[Theorem 
2]{BFP15}. They further observed that, it seems that most known 
algorithms on IQF1S would fail on the irregular instances, and proposed 
the complexity of such instances as an open question \cite[Sec. 1, 
Open Question]{BFP15}.

By the classical correspondence between quadratic forms and symmetric matrices, 
it is easy to see the equivalence between IQF1S and the isometry problem of tuples 
of 
symmetric matrices. Our algorithm for the latter problem then translates to a 
complete solution of IQF1S over 
finite fields of odd size, answering \cite[Sec. 1, Open Question]{BFP15} for such 
fields.
\begin{theorem}\label{thm:crypto}
Let $\F$ be a finite field of odd size. There exists a randomized polynomial-time 
algorithm that solves the Isomorphism of Quadratic Forms with One Secret problem 
over $\F$.
\end{theorem}

Furthermore, there has been a large body of works which aim to build 
public key cryptography schemes based on the hardness of solving systems of 
quadratic polynomials over finite fields. This approach is regarded as one 
candidate for 
post-quantum cryptography, in particular as a signature scheme \cite{NIST16}. We 
refer the reader to the thesis of Wolf 
\cite{Wol05} 
for an overview, and the recent article \cite{PCDY17} and references therein
for recent advances in this area. IQF1S 
and  
related problems play an important role in such schemes. As 
pointed out in \cite[Sec. 2.6.1]{Wol05}, though often not explicitly stated, it 
seems crucial to assume that 
IQF1S and related 
problems are difficult to ensure the security of these 
schemes. Theorem~\ref{thm:crypto} then suggests that the 
``one-secret'' versions of such schemes based on quadratic polynomials may not be 
secure.

\subsection{Group isomorphism problem}\label{subsubsec:gpi}

Group isomorphism problem (GpI) asks to 
decide 
whether two finite groups of order $n$ are isomorphic. It has been studied 
for several 
decades in both Computational Group Theory (CGT) and Theoretical Computer Science. 
The difficulty of this problem depends crucially on how we represent the groups in 
the algorithms. If the goal is to obtain an algorithm running in time $\poly(n)$, 
then we may assume that we have at our disposal the Cayley 
(multiplication) table of the group, as the Cayley table can be recovered from 
most 
reasonable 
models for computing with finite 
groups in time $\poly(n)$. 
Therefore, we restrict our discussion mostly to this very redundant 
model, 
%
which is meaningful mainly because we do not know a $\poly(n)$-time 
or even an $n^{o(\log 
n)}$-time algorithm \cite{Wil14} ($\log$ to the base $2$), despite that a simple 
$n^{\log n+O(1)}$-time 
algorithm has been known for decades \cite{FN70,Mil78}.
The past few years have witnessed a resurgence of activity on algorithms for this 
problem with worst-case analyses in terms of the group order;
%
we refer 
the 
reader to \cite{GQ17} which 
contains a survey of these algorithms.

%

It is long believed that $p$-groups (groups of a prime power order) form 
the bottleneck case for GpI. 
In fact, the decades-old quest for a polynomial-time algorithm has focused on 
class-$2$ $p$-groups, with little success. 
Even
if we restrict further to $p$-groups of class 2 and exponent $p$, the 
problem 
is still difficult.
Recently, some impressive progress on such $p$-groups was 
made on the CGT side, as seen in the works of Wilson, Brooksbank, and their 
collaborators \cite{Wil09,LW12,BMW15}. 

Most notably, a main result in 
\cite{BMW15} is a polynomial-time algorithm for $p$-groups of class $2$ and 
exponent $p$, when the commutator subgroup is of order $p^2$, in the model of 
quotients of permutation groups \cite{KL90}. This of course settles the same case 
in the 
Cayley table model. In fact, the same class of groups in the Cayley table model 
can be handled using one specific technique called the Pfaffian isomorphism test 
in \cite[Sec. 6.2]{BMW15}.
Still, despite 
all the progress, an efficient algorithm for $p$-groups of class $2$ and exponent 
$p$, with the commutator subgroup of order even $p^3$, was not 
known in the Cayley table model. Since we now have an efficient algorithm to test 
isometry of tuples of skew-symmetric matrices, the following result can be 
established.

\begin{theorem}\label{cor:p_group}
Let $p$ be an odd prime, and let two $p$-groups of class $2$ and exponent 
$p$ of order $p^\ell$, $G$ and $H$, be given by Cayley tables. If the 
commutator subgroup of $G$ is of order 
$p^{O(\sqrt{\ell})}$, then there exists a deterministic\footnote{The deterministic 
here is due to the last statement on derandomization of Theorem~\ref{thm:isometry} 
(1). That statement applies to the setting here, because the underlying field if 
$\F_p$ and our target is an algorithm with running time $p^{O(\sqrt{\ell})}$.} 
polynomial-time algorithm 
to 
test whether $G$ and $H$ are isomorphic. 
\end{theorem}

We explain how to obtain Theorem~\ref{cor:p_group} from our result. While the 
following 
reduction is well-known in CGT, we include it 
here for readers from other areas. Given a 
class $2$ and exponent $p$ $p$-group $G$, let $[G, G]$ denote its commutator 
subgroup. 
Due to the exponent $p$ and class $2$
condition, we have $G/[G,G]\cong \Z_p^n$ and $[G,G]\cong \Z_p^m$ for some $n$ and 
$m$ such that $n+m=\ell$. Fixing bases of $G/[G,G]$ and $[G,G]$, and taking the 
commutator bracket, we 
obtain a skew-symmetric bilinear map $b_G:\F_p^n\times \F_p^n\to \F_p^m$, 
represented by $\vecB\in S^{-1}(n, p)^m$. For $H$ to be isomorphic to $G$, it is 
necessary that $\dim_{\Z_p}(H/[H,H])=\dim_{\Z_p}(G/[G,G])$ and 
$\dim_{\Z_p}([H,H])=\dim_{\Z_p}([G,G])$, so by the same construction we obtain 
another $\vecC\in S^{-1}(n, p)^m$. We then need the following definition.
%
\begin{definition}\label{def:pseudo_isometry}
Given $\vecB=(B_1, \dots, 
B_m)$ and $\vecC=(C_1, \dots, C_m)$ from $S^\epsilon(n, \F)$, 
$\vecB$ and $\vecC$ are 
\emph{pseudo-isometric}, if there exists $X\in\GL(n, \F)$ such that $\langle 
X^tB_1X, \dots, X^tB_mX\rangle=\langle C_1, \dots, C_m\rangle$.
\end{definition}
The key connection then is Baer's correspondence, which, put in this context, 
gives that $G$ and $H$ are isomorphic if and only if $\vecB$ and $\vecC$ are 
pseudo-isometric \cite{Baer}.
By the condition that $m=O(\sqrt{\ell})$, we can enumerate all bases of $\vecC$ at 
a multiplicative cost of $p^{m^2}=p^{O(\ell)}$, and for each fixed basis, apply 
the algorithm for isometry testing. This gives 
Theorem~\ref{cor:p_group}.

As Brooksbank and Wilson have communicated to us, our algorithm may be 
useful in some models studied in CGT. Also, in 
multivariate cryptography, the 
problem Isomorphism of Quadratic Forms with Two Secrets (IQF2S) just asks to test 
the pseudo-isometry of tuples of symmetric matrices. 
Formally, the 
IQF2S problem 
asks to decide, given $\vecB, \vecC\in S^1(n, \F)$, whether they are 
pseudo-isometric.
Therefore a result analogous to Theorem~\ref{cor:p_group} can be obtained for 
IQF2S.

\subsection{Polynomial identity testing}\label{subsubsec:pit}

Fix $\epsilon\in\{1, -1\}$. Let us see
how to cast the $\epsilon$-symmetrization problem as an instance of the polynomial 
identity testing problem. Given $\vecB=(B_1, \dots, 
B_m)\in M(n, \F)^m$, there 
exist invertible matrices $A, D$ such that $\forall i\in[m]$, $AB_iD$ is 
$\epsilon$-symmetric if and only if $\forall i\in[m]$,  
$D^{-t}AB_i=D^{-t}(AB_iD)D^{-1}$ is 
$\epsilon$-symmetric. Therefore we can reduce to finding an invertible matrix $E$ 
such that $\forall i\in[m]$, $EB_i$ is $\epsilon$-symmetric. Suppose for now that 
$E$ is a matrix of variables. The equations
$\forall i\in[m], EB_i=\epsilon B_i^tE^t$ 
set up a system of linear forms in these variables. Let $C_1, \dots, C_\ell$ be a 
linear basis of the solution space, and $\cC$ be the matrix space $\langle 
C_1,\dots, C_\ell\rangle\leq M(n, \F)$. The problem then 
becomes to decide whether $\cC$
contains an invertible matrix. 
To decide whether a matrix space, given by a linear basis, contains only 
non-invertible 
matrix is known as the
symbolic determinant identity testing (SDIT) problem, which is equivalent to the 
polynomial identity testing (PIT) for weakly skew arithmetic circuits 
\cite{Tod92}\footnote{An 
arithmetic circuit is weakly skew if each product gate is of fan-in 2 and has at 
least one child such that the subcircuit rooted at it is separate from the other 
parts of the circuit \cite{Tod92,MP08}. The computation power 
of weakly skew circuit is known to be equivalent to the model of symbolic 
determinants, and between arithmetic formulas and arithmetic circuits.}. 

When $|\F|=\Omega(n)$, SDIT admits a randomized efficient algorithm via the 
Schwartz-Zippel lemma. To devise a deterministic efficient algorithm for SDIT is a 
major problem in algebraic complexity theory due to its implication to 
arithmetic circuit lower bounds. Specifically, in \cite{CIKK15} (building on 
\cite{ik2004}), Carmosino et al. 
show that such an algorithm implies the existence of a 
polynomial family such that its graph is in $\mathrm{NE}$, but it cannot be 
computed by polynomial-size arithmetic circuits. Such a lower bound is generally 
considered to be beyond current techniques, and would be recognized as a 
breakthrough if established. The research into PIT has 
received quite a lot of attention since early 2000's (see the surveys 
\cite{Sax09,sy2010,Sax13}).

Our algorithm for the $\epsilon$-symmetrization problem then provides a 
deterministic solution 
to this specific instance of SDIT. Our motivation to look at this problem at 
the first place was from the recent resolution of the non-commutative rank 
problem  by Garg et al. \cite{GGOW} and Ivanyos et al. 
\cite{IQS16a,IQS17}, and the intricate relation between the non-commutative rank 
problem 
and SDIT, which we explain 
below. 

A matrix 
space $\cB\leq M(n, \F)$ is non-singular, if $\cB$ contains an invertible matrix, 
and 
singular otherwise. SDIT then asks to decide whether a matrix space is 
singular. To obtain an arithmetic circuit lower bound 
via \cite{CIKK15}, it is actually enough to put SDIT in $\mathrm{NP}$, that is, to 
find a small witness  that helps to testify 
the singularity 
of singular matrix 
spaces. One such 
singularity witness, which is the reminiscent of the ``shrunk subset'' as in 
Hall's marriage theorem for bipartite graphs, and closely related to the linear 
matroid 
intersection problem \cite{Lovasz}, is the following. For $\cB\leq M(n, \F)$, 
$U\leq \F^n$ is a shrunk subspace of $\cB$, if 
$\dim(U)>\dim(\cB(U))$ where 
$\cB(U)=\langle B(U) : B\in \cB\rangle$. The decision 
version of the 
non-commutative rank problem then 
asks to decide whether $\cB$ 
has a shrunk subspace. Deterministic efficient algorithms for the 
non-commutative rank problem were 
recently devised in \cite{GGOW} (over $\Q$) and in \cite{IQS16a,IQS17} (over any 
field). 

A direct consequence of settling the non-commutative rank problem on SDIT 
is that we can restrict our attention to 
those singular matrix spaces without a shrunk subspace, which we call exceptional 
spaces. As described by Lov\'asz in \cite{Lovasz} (see also 
\cite{PrimitiveII,EH88}), the skew-symmetric structure naturally yields two 
families 
of exceptional spaces. 
To introduce them we need the following definition. Two matrix spaces $\cB, 
\cC\leq M(n, 
\F)$ 
are \emph{equivalent}, if there exist $A, D\in \GL(n, \F)$ such that $A\cB D=\cC$ 
(equal as subspaces). Note that whether a matrix space is singular is preserved by 
the equivalence relation. We now list the two families from \cite{Lovasz}.
\begin{enumerate}
\item[(1)] If $n$ is odd and $\cB\leq M(n, 
\F)$ is equivalent to a subspace in 
$S^{-1}(n, \F)$, then $\cB$ is singular, as every skew-symmetric matrix is of 
even rank. 
\item[(2)] Given $C_1, \dots, C_n\in S^{-1}(n, \F)$, let $\cC\leq M(n, \F)$ 
consist of 
all the 
matrices of the form $[C_1v, C_2v, \dots, C_nv]$ over $v\in \F^n$. Since 
$v^t[C_1v, C_2v, \dots, C_nv]=[v^tC_1v, v^tC_2v, \dots, 
v^tC_nv]=0$, $\cC$ is singular, and we call such $\cC$ a skew-symmetric induced 
matrix space. If $\cB$ is equivalent to a skew-symmetric induced 
matrix space, then $\cB$ 
is singular as well. Note that w.l.o.g. we can assume that $\cB$ is a subspace of 
$M(n, \F)$ of dimension $n$.
\end{enumerate}

These two families of exceptional matrix spaces can be deterministically 
recognized as follows.
\begin{theorem}\label{thm:pit}
Let $\F$ be a field of characteristic not $2$. Given $\cB=\langle B_1, \dots, 
B_m\rangle\leq M(n, \F)^m$, there exists a deterministic 
polynomial-time algorithm that decides whether $\cB$ is equivalent to a subspace 
in $S^{-1}(n, \F)$, or a skew-symmetric induced matrix space. 
\end{theorem}
We explain how Theorem~\ref{thm:pit} follows from our $\epsilon$-symmetrization 
algorithm. 
The case (1) is straightforward: apply the 
skew-symmetrization algorithm to the given
linear basis of $\cB$. In case (2), suppose $B_i=[b_{i,1}, \dots, b_{i, n}]$ where 
$b_{i, j}\in \F^n$, $j\in[n]$ are the 
columns of $B_i$. Following an observation of Lov\'asz in \cite{Lovasz}, 
construct 
$B'_i=[b_{1,i}, \dots, b_{n,i}]$ for $i\in[n]$. It can be verified that $\cB$ is 
equivalent 
to some $\cC$ of the form described in (2) if and only if $\cB'=\langle B'_1, 
\dots, B'_n\rangle$ is equivalent to a subspace in $S^{-1}(n, \F)$. We can then 
apply the skew-symmetrization algorithm to $(B'_1, \dots, B'_n)$ to conclude.

\subsection{Results and techniques}\label{subsec:result}

\paragraph{Statement of the results.} We first
define three equivalence relations for 
matrix tuples. 
\begin{definition}\label{def:three}
Let 
$\vecB=(B_1, \dots, B_m), 
\vecC=(C_1, \dots, C_m)\in M(n, \F)^m$. $\vecB$ and 
$\vecC$ 
are \emph{conjugate}, if $\exists A\in \GL(n, \F)$, such that $A\vecB
=\vecC A$. They are \emph{equivalent}, if $\exists A, D\in \GL(n, \F)$, such that 
$A\vecB 
=\vecC D$. They are \emph{isometric}, denoted as $\vecB\sim \vecC$, if $\exists 
A\in \GL(n, \F)$, such that 
$A^t\vecB 
A=\vecC$; such an $A$ is called an isometry from $\vecB$ to $\vecC$.
\end{definition}

We show that testing whether two $\epsilon$-symmetric matrix tuples are isometric 
can be solved efficiently over $\F_q$ with $q$ odd, $\R$, and $\C$. Note that the 
algorithm 
for $\F_q$ is probabilistic.

\begin{theorem}\label{thm:isometry}
\begin{enumerate}
\item (Finite fields of odd size) Given $\vecB, \vecC\in 
S^\epsilon(n, q)^m$ with $q$ odd, there exists a randomized polynomial-time 
algorithm that 
decides whether $\vecB$ and $\vecC$ are isometric. If $\vecB$ and $\vecC$ are  
isometric, the algorithm also computes an explicit isometry in $\GL(n, q)$. This 
algorithm can be derandomized at the price of running in time $\poly(n, m, 
\log q, p)$ where $p=\fdchar(\F_q)$.
\item (The real field $\R$) Let $\E\subseteq \R$ be a number field. Given $\vecB, 
\vecC\in S^\epsilon(n, \E)^m$, there exists a deterministic polynomial-time 
algorithm 
that decides whether $\vecB$ and $\vecC$ are isometric over some number field $\K$ 
such that $\E\subseteq\K\subseteq \R$. If $\vecB$ and $\vecC$ are indeed 
isometric, the algorithm also computes an explicit isometry, represented as a 
product of matrices, where each matrix is over some extension field of $\E$ of 
extension degree $\poly(n, m)$. 
\item (The complex field $\C$) Let $\E$ be a number field. Given $\vecB, 
\vecC\in S^\epsilon(n, \E)^m$, there exists a deterministic polynomial-time 
algorithm 
that decides whether $\vecB$ and $\vecC$ are isometric over some number field $\K$ 
such that $\E\subseteq\K$. If $\vecB$ and $\vecC$ are indeed 
isometric, the algorithm also computes an explicit isometry, represented as a 
product of matrices, where each matrix is over some extension field of $\E$ of 
extension degree $\poly(n, m)$.  
\end{enumerate}
\end{theorem}


%

We call $\vecB\in M(n, \F)^m$ 
\emph{$\epsilon$-symmetrizable}, if $\vecB$ is equivalent to a tuple of 
$\epsilon$-symmetric matrices. Our second main result concerns the problem of 
testing whether a matrix tuple is $\epsilon$-symmetrizable. 
\begin{theorem}\label{thm:sym}
Let $\F$ be a field of characteristic not $2$. 
Given $\vecB\in M(n, \F)^m$, there exists a deterministic 
algorithm that decides whether $\vecB$ is $\epsilon$-symmetrizable, and if it is, 
computes $A, 
D\in\GL(n, \F)$ such that $A\vecB D\in S^\epsilon(n, \F)^m$. The 
algorithm uses 
polynomially many 
arithmetic operations. Over a number field the final data as well
as all the intermediate data have size polynomial in the input data size, 
hence the algorithm runs in polynomial time.
\end{theorem}


%


\paragraph{Two key ingredients.}
Let us first review the concept of $*$-algebras, and see how to get a 
$*$-algebra from a tuple of $\epsilon$-symmetric matrices. Recall that, a 
$*$-algebra $A$ is an algebra with $*:A\to A$ being 
an anti-automorphism of order at most $2$. $*$-algebras have been studied since 
1930's \cite{Alb39} (see \cite{Lew06} for a recent survey). 
Let $M(n, \F)^{op}$ be the 
opposite full matrix algebra, which is the ring consisting of all matrices in 
$M(n, \F)$ with the multiplication $\circ$ as $A\circ B=BA$. $*$-algebras arise 
from $\epsilon$-symmetric matrix tuples by considering 
the \emph{adjoint algebra} of $\vecB\in S^\epsilon(n, \F)^m$, which consists of 
$\{(A, 
D)\in M(n, 
\F)^{op} \oplus M(n, \F) | A^t\vecB=\vecB D\}$, with a natural involution $*$ 
as 
$(A, 
D)^*=(D, A)$. 
%

We then turn to the module isomorphism problem (MI). Given $\vecB, \vecC\in M(n, 
\F)^m$, MI asks if $\vecB$ and $\vecC$ are conjugate. 
This problem is termed as module isomorphism, as the matrix
tuple $\vecB=(B_1, \dots, B_m)$ can be viewed as a linear representation of a 
finitely 
generated 
algebra generated by $m$ elements. Two deterministic polynomial-time algorithms 
for MI have been devised in \cite{CIK97,IKS10} and \cite{BL08}. Note that MI may 
also 
be cast as an instance of the polynomial identity testing problem like the  
$\epsilon$-symmetrization problem. 

%
%

\paragraph{More comparison with previous works.}
Some comparisons with previous works were already stated in 
Section~\ref{subsubsec:crypto} and~\ref{subsubsec:gpi}. We now add some more 
details on the technical side. 
In Section~\ref{subsubsec:crypto}, we mentioned the work of 
Berthomieu et al.~\cite{BFP15} which solves the IQF1S 
possibly over an extension field, for regular instances and large 
enough fields. Here we seek ``rational'' solutions (i.~e.~those
over the given base field) in the finite case and seek solutions
over a real extension field. 
An interesting observation is that the algorithm of 
Berthomieu et al.~may be cast as working 
with a $*$-algebra, but in a much restricted setting. We 
explain this in detail in Appendix~\ref{app:compare}.
In Section~\ref{subsubsec:gpi}, we described how our result, when applied to 
$p$-group 
isomorphism, compares to 
the result of Brooksbank et al.~\cite{BMW15}. The relevant technique there, 
called 
the 
Pfaffian isomorphism test \cite[Sec. 6.2]{BMW15}, is completely different from 
ours, 
and seems quite restricted to pairs of skew-symmetric matrices. 


The work
\cite{BW12} by Brooksbank and 
Wilson is the most important precursor to our Theorem~\ref{thm:isometry}. 
In \cite{BW12}, the main result, rephrased in our 
setting, is an efficient 
algorithm that, given $\vecB\in S^\epsilon(n, q)^m$ with $q$ odd, computes a 
generating set for the group $\{X\in\GL(n, q) \mid  X^t\vecB X=\vecB\}$. This is 
exactly the 
``automorphism version'' of the isometry problem. 
However, unlike many other 
isomorphism problems, the isometry problem is not known to reduce to this 
automorphism version. This is similar to the module isomorphism problem: the 
automorphism version of 
MI asks to compute a generating set of the unit group in a matrix algebra, which 
was solved in \cite{BO08}. The ideas and the techniques for the unit group 
computation in \cite{BO08} and for MI in \cite{CIK97,IKS10,BL08} are totally 
different. So Theorem~\ref{thm:isometry} cannot be 
easily deduced as a corollary from \cite{BW12}.

\paragraph{Generalizations of the main results.} 
Theorem~\ref{thm:isometry} can be generalized to the following setting.
Following \cite{BW12}, for an linear automorphism $\theta\in \GL(W)$
we call a
bilinear map over a field $\F$, $b:V\times V\to W$ $\theta$-\emph{Hermitian}, 
if 
for all $u, v\in V$, $b(u, v)=\theta(b(v, u))$. Obviously, nontrivial
Hermitian maps exist only if $\theta^2$ is the identity.
Hermitian 
bilinear maps subsume symmetric bilinear maps ($\theta$ being the identity matrix) 
and skew-symmetric bilinear maps ($\theta$ being $-1$ times the identity matrix). 
It allows for (after 
fixing bases of $V$ and $W$) a tuple of mixed symmetric and skew-symmetric 
matrices. In fact, by a change of basis of $W$, we may always assume that $\theta$ 
is a diagonal matrix with $1$ and $-1$'s on the diagonal and 
in our arguments and algorithms 
we only need
the replace $\epsilon$ by a tuple $(\epsilon_1,\ldots,\epsilon_m)$
and equations of type $B_i^t=\epsilon B_i$ by $B_i^t=\epsilon_i B_i$.
 Furthermore, the concept captures 
Hermitian forms by \cite[Sec. 3.1]{BW12}: for a Hermitian form $b:V\times V \to 
\F_{q^2}$ where $V\cong \F_{q^2}^n$, we can represent it as a pair of bilinear 
forms over $\F_q$, $b_1, b_2: V'\times V'\to \F_q$ where $V'\cong \F_q^{2n}$, and  
$\theta\in \GL(2, q)$ corresponds to the field involution 
$\alpha\to \alpha^q$ for 
$\alpha\in \F_{q^2}$. Hermitian complex or quaternionic matrices are also
included: assume that $D$ is a finite 
dimensional division algebra over $\F$ with involution 
$\overline{\,\cdot\,}:D \rightarrow D$, such that $\F$ coincides
with the subfield of the center of $D$ consisting of the elements
fixed by $\overline{\,\cdot\,}$. Then the map $*$ sending a matrix to
the transpose of its elementwise $\overline{\,\cdot\,}$-conjugate 
is an involution on $M(n,D)$, and the matrices invariant under $*$
are called $*$-Hermitian. Indeed, let $d$ be the dimension
of $D$ over $\F$. Then we can interpret $D$ and $D^n$ as vector spaces
of dimension $d$ resp.~$dn$ over $\F$, and a matrix in $M(n,D)$ as 
an $\F$-bilinear map from $D^n\times D^n$ to $D$. Then $*$-Hermitian
matrices are interpreted as Hermitian bilinear maps
for $\overline{\,\cdot\,}$. (Naturally, an $m$-tuple of $*$-Hermitian
matrices become a Hermitian map from $D^n\times D^n$ to $D^{m}$.)

Interestingly, Theorem~\ref{thm:isometry} allows us to solve the isometry problem 
for a tuple of 
arbitrary matrices over $\F_q$ with $q$ odd, $\R$, or $\C$. 
Given $\vecB, \vecC\in M(n, \F)^m$, we can construct 
$\vecB'=(\frac{1}{2}(B_1+B_1^t), \dots, 
\frac{1}{2}(B_m+B_m^t), 
\frac{1}{2}(B_1-B_1^t), \dots, \frac{1}{2}(B_1-B_1^t))$, and similarly 
$\vecC'$. Here we use the fact that we work over fields of characteristic not $2$. 
Then it is easy to verify that $\vecB\sim \vecC$ if and only if 
$\vecB'\sim \vecC'$. 
Indeed, if $A\in \GL(n, \F)$ 
satisfies that $A^tB_iA=C_i$, 
then $A$ also satisfies that $A^t(\frac{1}{2}(B_i\pm 
B_i^t))A=\frac{1}{2}(A^tB_iA\pm A^tB_i^tA)=\frac{1}{2}(C_i\pm C_i^t)$. On the 
other hand, if $A^t(\frac{1}{2}(B_i+B_i^t))A=\frac{1}{2}(C_i+ C_i^t)$ and 
$A^t(\frac{1}{2}(B_i-B_i^t))A=\frac{1}{2}(C_i- C_i^t)$, summing these two we get 
that $A^tB_iA=C_i$. 
Combining with the observation from the last paragraph, we 
have the following. 
\begin{corollary} 
The statement of 
Theorem~\ref{thm:isometry} holds for $\vecB, \vecC\in M(n, \F_q)^m$, $M(n, 
\E)^m$ with a number field $\E\subseteq \R$, or $M(n, \E)^m$ with a number field 
$\E$. 
\end{corollary}

Theorem~\ref{thm:sym} can also be generalized to
transforming bilinear maps to $\theta$-Hermitian ones, including
the case of tuples of complex and quaternionic matrices.

\paragraph{Some open problems.} There are two immediate open problems left. 

The 
first one is to extend both of our 
results to 
fields of characteristic $2$. While presenting the algorithm for the isometry 
problem in 
Section~\ref{sec:iso}, we indicate explicitly in each step whether the 
characteristic not $2$ is required, and one may want to examine those steps where 
the characteristic not $2$ condition is crucial. For the $\epsilon$-symmetrization 
problem, one may want to start with examining the key lemma, Lemma~\ref{lem:key}, 
in the setting of characteristic-$2$ fields. 

The second one is to solve the 
isometry test 
problem 
over a number field without going to extension fields. To extend our current 
approach to deal with the second problem involves certain 
number-theoretic 
obstacles even over $\Q$. 
Namely, our present method relies on 
representing
a simple algebra explicitly as a full matrix 
algebra over a division ring, but there is a randomized reduction 
from factoring squarefree integers 
to this task for a central simple algebra of dimension $4$ over $\Q$ 
assuming the Generalized Riemann Hypothesis \cite{Ron87}. Even
deciding whether a four dimensional non-commutative simple algebra
over $\Q$ is isomorphic to $M(2,\Q)$ is equivalent to
deciding quadratic residuosity modulo composite numbers.
This kind of obstacles appears to be inherent: a ternary quadratic form
over $\Q$ is isotropic 
if and only if an associated non-commutative simple 
algebra of dimension four over $\Q$ is isomorphic to $M(2,\Q)$.
Now consider an indefinite symmetric $3$ by $3$ matrix 
$B$ with rational entries having determinant $d$.
Then the ternary quadratic form with Gram matrix $B$
is either anisotropic or isometric to the form
having matrix 
$$\begin{pmatrix}
0 & 1 & 0\\
1 & 0 & 0 \\
0 & 0 & -d
\end{pmatrix}.$$
Thus over $\Q$, the isometry problem a single 
ternary quadratic 
form is at least as hard as deciding whether an algebra is
isomorphic to $M(2,\Q)$. Actually, there is a randomized
polynomial time reduction from testing whether a simple
algebra over a number field $\F$ is isomorphic with a full matrix
algebra over $\F$ to factoring integers, see~\cite{Ron92} 
and~\cite{IR93}.
However, for the constructive
version of isomorphisms with full matrix algebras
such a reduction is only known for the case $M(n,K)$ where
$n$ is bounded by a constant, and $K$ is from a finite collection
of number fields~\cite{IRS12}.
Therefore, to determine the relation between the complexity
of the isometry problem and that of factoring,   
it might be useful to devise an alternative approach 
which gets around constructing explicit isomorphisms 
with full matrix algebras. 

\paragraph{Future directions.} Given Theorem~\ref{thm:isometry}, the next target 
is of course to study IQF2S and isomorphism testing of $p$-groups of class $2$ and 
exponent $p$. 
For these two problems, the first goal would be to design, for $\vecB\in 
S^\epsilon(n, 
q)^m$, an algorithm in time $q^{O(n+m)}$. In the context of $p$-groups of class 
$2$ and exponent $p$, this amounts to solve isomorphism testing for this group 
class in time polynomial in the group order, which seems a difficult problem 
already. By Theorem~\ref{thm:isometry}, this target seems most 
difficult when $m$ and $n$ are comparable, say $m=n$. 
One idea may be to reduce to the parameters $m'$ and $n'$ such that 
$m'=O(n^{1/2})$ and $n'=\poly(n)$, so that we can use Theorem~\ref{thm:isometry} 
to get an algorithm in time $q^{O(n)}$. It is also noteworthy that recently, 
Yinan Li 
and the second author devised an algorithm for $m=\Theta(n)$ in 
\emph{average-case} 
time $q^{O(n)}$ \cite{LQ17}; 
the average-case analysis is done in a random model 
for linear spaces of skew-symmetric matrices over finite fields, that can be 
viewed as a linear 
algebraic analogue of the Erd\H{o}s-R\'enyi model for random graphs. 

Theorem~\ref{thm:pit} represents a natural step in the direction for derandomizing 
SDIT set up by the resolution of the non-commutative rank problem 
\cite{GGOW,IQS16a,IQS17}. While most research activities on PIT and SDIT put 
constraints on the structural properties of the arithmetic circuits 
\cite{Sax09,sy2010,Sax13}, 
this direction puts constraints on the singularity witnesses which are inspired by 
geometric considerations \cite{EH88} and/or combinatorial considerations 
\cite{Lovasz}. At present, we are not aware of an explicit connection between 
these two different styles of constraints. It is an interesting question as to 
whether these geometric and/or combinatorial considerations can be made more 
systematic to yield a formal strategy to attack SDIT.

\paragraph{Organization of the article.} In 
Section~\ref{sec:prel}, we present certain preliminaries, including those 
structural results of $*$-algebras that are relevant to us. In 
Sections~\ref{sec:iso}, we give a 
detailed description of the algorithm for 
Theorems~\ref{thm:isometry}. In Section~\ref{sec:sym}, we show that for the 
$\epsilon$-symmetrization problem, how to handle the cases when the 
Jacobson radical is not known to be efficiently computable, or the field is too 
small, finishing the proof of Theorem~\ref{thm:sym}.

\section{Preliminaries}\label{sec:prel}

\paragraph{Notation.} For $n\in\N$, $[n]:=\{1, \dots, n\}$. For a field $\F$, 
$\fdchar(\F)$ denotes the 
characteristic of $\F$. $\vzero$ 
is the zero vector. For $B\in M(n, \F)$, $i, j\in[n]$, $S, T\subseteq [n]$, $B(i, 
j)$ is the $(i,j)$th entry of 
$B$, $B(S, T)$ is the submatrix indexed by row indices in $S$ 
and column indices in $T$. We use $I_n$ to denote the $n\times n$ 
identity matrix, and 
$\langle \cdot \rangle$ to denote the linear span. 
The vector space 
$\F^n$ consists of length-$n$ \emph{column} vectors over $\F$.

Given a quadratic field extension $\F/\F'$, for $\alpha\in \F$, its conjugation
$\overline{\alpha}$ is the image of $\alpha$ under the quadratic field 
involution.  When $\F=\C$ and $\F'=\R$ this is simply the complex conjugation. We 
use $\qH$ to 
denote the quaternion division algebra over $\R$, and $i, j, k$ be the 
fundamental quaternion units. For $\alpha=a+bi+cj+dk\in \qH$, 
its conjugation, denoted also by $\overline{\alpha}$, is $a-bi-cj+dk$.
Given $A\in M(n, \F)$ or $M(n, \qH)$, $\overline{A}$ denotes 
the 
matrix obtained by applying conjugation to every entry of $A$. For 
$\epsilon\in\{1, -1\}$ and $A\in M(n, \F)$ or $M(n, \qH)$, $A$ is 
$\epsilon$-Hermitian, if $\overline{A}^t=\epsilon A$.
%

We will also meet matrices over division rings, and therefore, for a division ring 
$D$, the notation $M(n, 
D)$ (for the full $n\times n$ matrix ring over $D$) and $\GL(n, D)$ (for the group 
of units 
in $M(n, D)$).

\paragraph{Representation of fields and field extensions.} For the isometry 
problem, we 
assume the input matrices are over a field $\E$ such that 
$\E$ is a finite extension of its prime field $\F$ (so $\F$ is either a field of 
prime order or $\Q$). Therefore $\E$ is a 
finite-dimensional algebra over $\F$. If $\dim_\F(\E)=d$, then $\E$ is the 
extension of $\F$ by a single generating element $\alpha$, so $\E$ can be 
represented by the minimal polynomial of $\alpha$ over $\F$, together with an 
isolating interval for $\alpha$ in the case of
$\R$, or an isolating rectangle for $\alpha$ in the case of $\C$.
When we say that we work over $\R$ (resp. $\C$), the input is given as over a 
number 
field $\E\subseteq \R$ (resp. $\E\subseteq \C$). The algorithm is then allowed to 
work with 
extension fields of $\E$ in $\R$ (resp. $\C$), as long as the extension degrees 
are polynomially 
bounded. On the other hand, if we say that we work with a number field, we 
usually assume that we do not need to work with further extensions. 

For the $\epsilon$-symmetrization problem, we work with the arithmetic model, 
namely the fundamental steps are basic field operations, and the complexity is 
determined by counting the number of such basic operations. Furthermore, over 
number 
fields we are also concerned with the bit complexity. So when we say that some 
procedure works over any field, we mean that the procedure uses polynomially 
arithmetic operations, and when over number fields, $\R$ or $\C$, the bit 
complexity is also polynomial. 

\paragraph{Tuples of matrices.} A matrix tuple is an element in $M(n, \F)^m$, and 
an $\epsilon$-symmetric matrix tuple is an element in $S^\epsilon(n, \F)^m$. We 
will mostly use $\vecB$, $\vecC$ to denote matrix tuples. Given $\vecB=(B_1, 
\dots, B_m)\in M(n, \F)^m$, define its kernel, $\ker(\vecB)$, as 
$\cap_{i\in[m]}\ker(B_i)$, and its image, $\im(\vecB)$, as $\langle 
\cup_{i\in[m]}\im(B_i)\rangle$. $\vecB\in M(n, \F)^m$ is \emph{non-degenerate}, if 
$\ker(\vecB)=\vzero$, and $\im(\vecB)=\F^n$. For $\vecB\in S^\epsilon(n, \F)^m$, 
due to the 
$\epsilon$-symmetric condition, it can be verified easily that $\im(\vecB)=\{ v\in 
\F^n : \forall u\in \ker(\vecB), u^tv=0\}$. So $\vecB\in S^\epsilon(n, \F)^m$ is 
non-degenerate if and only if $\ker(\vecB)=\vzero$.

Given $\vecB=(B_1, \dots, B_m)\in M(n, 
\F)^m$, $\vecB^t=(B_1^t, \dots, B_m^t)$. Given $\alpha\in \F$, 
$\alpha\vecB=(\alpha B_1, \dots, \alpha B_m)$. So for $\vecB\in S^\epsilon(n, 
\F)$, $\vecB^t=\epsilon\vecB$. Given $A, D\in M(n, \F)$, 
$A\vecB D=(AB_1D, \dots, AB_mD)$. 
Given $\vecB, \vecC\in M(n, \F)^m$, $\vecB$ and $\vecC$ are 
\emph{conjugate}, if there exists $A\in \GL(n, \F)$ such that $A\vecB=\vecC A$. 
$\vecB$ and $\vecC$ are \emph{equivalent}, if there exists $A, D\in\GL(n, 
\F)$ such that $A\vecB=\vecC D$. 
The classical module isomorphism problem asks to decide whether $\vecB$ and 
$\vecC$ are conjugate. 
\begin{theorem}[{\cite{CIK97,BL08,IKS10}}]\label{thm:mi}
Let $\vecB$ and $\vecC$ be from $M(n, \F)^m$. There exists a deterministic 
algorithm that 
decide whether $\vecB$ and $\vecC$ are conjugate. The algorithm uses polynomially 
many arithmetic operations. Over number fields 
the bit 
complexity of 
the algorithm is also polynomial. 
\end{theorem}

\paragraph{Structure of algebras.}
The proofs in this paper rely heavily on structure of
finite dimensional algebras, so we recall in nutshell some 
of the most important notions and facts from their theory. Classical references 
include \cite{Pierce}, and a concise introduction can be found in \cite[Sec. 
5]{AB95}.
All the algebras we consider are finite dimensional associative algebras
over some field $\F$. An ideal is a linear subspace
of $\cA$ closed under multiplication by elements of $\cA$, both
from the left and from the right. Left ideals are subspaces
closed under multiplication by elements of $\cA$ from the left, right
ideals are defined analogously. In this context, an ideal or, more
generally, a subalgebra $S$ is nilpotent when $S^n$, the subspace spanned
by products of length $n$ of element from $S$ are zero for some $n$.
An algebra $\cA$ has a largest nilpotent ideal
$\Rad(\cA)$, called the Jacobson radical, 
also simply referred to 
as the 
radical in this paper. We 
will make use of an 
alternative
characterization of the radical, namely, it is the intersection of the maximal 
right ideals (or the intersection of maximal left ideals).

An algebra 
is simple when it contains no proper and 
non-trivial 
(two-sided)
ideals. 
A semisimple algebra is 
isomorphic to a direct sum
of simple algebras. 
The factor algebra $\cA/\Rad(\cA)$
is semisimple. In a finite dimensional
algebra over a field, every nonzero element is either a unit (i.,e.,
has a multiplicative inverse) or a zero divisor (can be multiplied
by nonzero elements from both sides to obtain zero). In a division
algebra, also known as a skewfield, every nonzero element is a unit. 
A simple algebra is isomorphic to a full matrix algebra 
over a division algebra \cite[Theorem 17 on pp.129]{AB95}. 
Over finite fields all the division 
algebras are actually commutative, 
or in other words, they are 
fields; this is known as Wedderburn's little theorem.
Over an algebraically closed field there is even only one division algebra, that is
the base field itself \cite[Lemma 14 on pp. 127]{AB95}.
The structural results summarized above are also known as Wedderburn's theory, and 
a concise introduction can be found in \cite[Sec. 5]{AB95}.

An idempotent is a nonzero element $e$ with $e^2=e$. A semisimple
algebra necessarily contains at least one idempotent: the identity element.
Non-nilpotent algebras also contain idempotents (but not necessarily
identity elements). 
This follows 
from the following fact.
\begin{fact}\label{fact:non-nil}
Let
$\cA$ be a non-nilpotent algebra over a field $\F$. Every basis of $\cA$ contains 
a 
non-nilpotent element.
\end{fact}
\begin{proof}
Let $\K$ be the algebraic closure
of $\F$. Observe that $\cA$, as $=\F\otimes_\F\cA$, is embedded into
$\K\otimes \cA=:\overline \cA$. Then $\overline \cA$ is a
non-nilpotent $\K$-algebra and hence it has a full matrix algebra
as a factor. The image of an $\F$-basis of $\cA$ under the composition
of the embedding of $\cA$ into $\overline \cA$ with the natural projection
to this factor gives a system that spans a full matrix algebra over
$\K$. Now observe that a full matrix algebra cannot be spanned by
nilpotent matrices: nilpotent matrices have zero traces but there
exist matrices with nonzero trace even in positive characteristic. It follows that 
this $\F$-basis must contain at least one non-nilpotent element.
\end{proof}
The proof of Fact~\ref{fact:non-nil} shows that it is easy to find a non-nilpotent 
element in
a non-nilpotent algebra. 
Back to our task of locating an idempotent, let 
$y$ be a 
non-nilpotent element. Then the 
(commutative) subalgebra generated by $x=y^n$ for sufficiently large $n$ 
(say $n=\dim\cA$) has an identity element $e$, which is necessarily idempotent. 
To see this, note that 
the action of $y$ by left multiplication on the vector 
space $\cA$ yields the Fitting decomposition $\cA_0\oplus \cA_1$, such that 
$\cA_0=\ker(y^n)$ and $\cA_1=\im(y^n)$ for a large enough $n$. Consider the 
restriction of $y^n$ on $\cA_1$; its characteristic polynomial $f$ has a nonzero 
constant term $\alpha$. Then $(f(y^n)-\alpha)/\alpha$ is 
an element of the subalgebra generated by $y^n$ that gives an identity $e$ on 
$\cA_1$. Now observe that $y^n\in \cA_1$, so indeed $ey^{nk}=y^{nk}$ for any $k\in 
\N$. 
While the above argument shows the existence of $e$, a more straightforward way to 
compute this $e$ would be to express 
$e$ as a linear 
combination
of powers of $x$ whose coefficients are variables. Then $e$ can be computed in 
polynomial time, by solving a 
system of linear equations expressing the condition $ex=x$. 

A matrix representation of an algebra $\cA$ is 
a homomorphism of $\cA$ into a matrix algebra.
There is a straightforward
linear representation over $\F$ at hand, the so-called
left regular representation,  as follows. Let $V(\cA)$ be the vector space 
supporting the algebra $\cA$. Then $a\in \cA$ naturally acts as a linear
map on $V(\cA)$ as $\ell_a$ by sending $x\in V(\cA)$ to $ax$. The properties of
the algebra operations (most notably, though not exclusively,
 associativity of multiplication) ensure that $\ell:\cA\to\hom(V(\cA),V(\cA))$ by 
 sending $a$ to $\ell_a$ 
 is a 
homomorphism from $\cA$
into the algebra of $\F$-linear transformations of $\cA$. It is
an embedding when $\cA$ has an identity element.
We remark that image of $\ell_a$ is the right ideal $a\cA$ generated by
$a$, while its kernel is the right annihilator 
$\Ann_r(a)=\{x\in \cA:ax=0\}$ of $\cA$. It is straightforward
to check that $\Ann_r(a)$ is also a right ideal.

\paragraph{Structure of $*$-algebras.} We collect basic facts about 
$*$-algebras here. A classical reference for $*$-algebras is Albert's book 
\cite{Alb39}.
Fix a field $\F$, and let $\cA$ be 
an $\F$-algebra, e.g. an algebra over $\F$. Given an 
anti-automorphism $*:\cA\to\cA$ of order at most 2, $(\cA, 
*)$ is termed as a $*$-algebra. We will always assume that for an $\F$-algebra 
$\cA$, $*$ fixes $\F$, that is $\alpha^*=\alpha$ for $\alpha\in \F$.  An element 
$a\in \cA$ is \emph{$*$-symmetric} if 
$a^*=a$, and \emph{$*$-unitary} if $a^*a=1$. A 
$*$-homomorphism between $(\cA, 
*)$ and $(\cA', \circ)$ is an algebra homomorphism $\phi:\cA\to\cA'$ such that 
$\phi(a^*)=\phi(a)^\circ$. An ideal 
$I\subseteq \cA$ is an $*$-ideal, if $I^*=I$. 
The Jacobson radical of $\cA$, 
denoted as $\Rad(\cA)$, is the largest nilpotent ideal of $\cA$ as an 
$\F$-algebra. 
It is straightforward to verify that $\Rad(\cA)$ is a $*$-ideal. 
A $*$-algebra is $*$-simple, if it does not contain non-trivial $*$-ideals. 
Note that for a $*$-algebra $(S, *)$, if $S$ is simple, then it must 
be $*$-simple.
The semisimple $\cA/\Rad(\cA)$, with the induced 
involution (again denoted as $*$), is $*$-isomorphic to $(S_1, *)\oplus (S_2, *) 
\oplus \dots \oplus 
(S_k, *)$, where each $(S_i, *)$ is a $*$-simple algebra. 


A $*$-simple algebra $(S, *)$ over $\F$ falls into two categories. 
Either $S$ is a simple 
algebra, or $S$ is a direct sum of two anti-isomorphic simple algebras with $*$ 
interchanging the two summands \cite[Chap. X.3]{Alb39}. We shall refer to the 
latter as \emph{exchange type}, and its structure is easy to describe: an 
exchange-type 
$*$-simple algebra 
$(S, *)$ is 
$*$-isomorphic to $(M(n, D)\oplus M(n, D)^{op}, \circ)$, where $\circ$ is an 
involution 
sending $(A, B)$ to $(\phi^{-1}(B), \phi(A))$ for some algebra automorphism $\phi$ 
of 
$M(n, D)$.

When $S$ is simple, a general 
result regarding the possible forms of involutions is \cite[Chap. X.4, Theorem 
11]{Alb39}. 
We can explicitly list these forms for $\F_q$ with $q$ odd, $\R$, and $\C$ as 
follows. 

Over $\F_q$ with $q$ odd, finite simple $*$-algebras are classified as follows 
(see also \cite[Sec. 3.3]{BW12}). To start with, recall that a finite simple 
algebra $S$
over $\F_q$ is isomorphic to $M(n, \F_{q'})$ where $\F_{q'}$ is an extension field 
of $\F_q$. So without loss of generality we may 
assume $S=M(n, \F_{q'})$. Then any 
involution $*$ on $M(n, \F_{q'})$ is in one of the following forms. 
\begin{description}
\item[Orthogonal type] For $X\in 
M(n, \F_{q'})$, $X^*=A^{-1}X^tA$ for some $A\in\GL(n, \F_{q'})$, $A=A^t$.
\item[Symplectic type] For $X\in 
M(n, \F_{q'})$, $X^*=A^{-1}X^tA$ for some $A\in\GL(n, \F_{q'})$, $A=-A^t$.
\item[Hermitian type] $\F_{q'}$ is a quadratic extension of a subfield $\F_{q''}$. 
For $X\in M(n, 
\F_{q'})$, $X^*=A^{-1}\overline{X}^tA$ for some $A\in\GL(n, \F_{q'})$, 
$\overline{A}^t=A$.
\end{description} 

Over $\R$, finite simple $*$-algebras are classified as follows (see also 
\cite{Lew77}). To start with, recall that, by a theorem of Frobenius (see 
e.g. \cite{Pal68}), a finite simple algebra $S$ over 
$\R$ is isomorphic to either $M(n, \R)$, $M(n, \C)$, or $M(n, \qH)$. So without 
loss of generality we may assume $S$ is one of the above. Then any involution $*$ 
on $S$ is in one of the following forms. Note that each type corresponds to 
a classical group as in \cite{Wey97}.
\begin{description}
\item[Orthogonal type] $S=M(n, \R)$. For $X\in M(n, \R)$, $X^*=A^{-1}X^tA$, 
$A\in\GL(n, \R)$, $A=A^t$.
\item[Symplectic type] $S=M(n, \R)$. For $X\in M(n, \R)$, $X^*=A^{-1}X^tA$, $A\in 
\GL(n, \R)$, $A=-A^t$.
\item[Complex orthogonal type] $S=M(n, \C)$. For $X\in M(n, \C)$, 
$X^*=A^{-1}X^tA$, $A\in \GL(n, \C)$, $A=A^t$.
\item[Complex symplectic type] $S=M(n, \C)$. For $X\in M(n, \C)$, 
$X^*=A^{-1}X^tA$, $A\in \GL(n, \C)$, $A=-A^t$.
\item[Unitary type] $S=M(n, \C)$. For $X\in M(n, \C)$, 
$X^*=A^{-1}\overline{X}^tA$, $A\in \GL(n, \C)$, $A=\overline{A}^t$.
\item[Quaternion unitary type] $S=M(n, \qH)$. For $X\in M(n, \qH)$, 
$X^*=A^{-1}\overline{X}^tA$, $A\in \GL(n, \qH)$, $A=\overline{A}^t$.
\item[Quaternion orthogonal type] $S=M(n, \qH)$. For $X\in M(n, \qH)$, 
$X^*=A^{-1}\overline{X}^tA$, $A\in\GL(n, \qH)$, $A=-\overline{A}^t$.
\end{description}
On $\C$, $\overline{\,\cdot \,}$ denotes the standard conjugation 
$a+bi\mapsto a-bi$, 
while on $\qH$ it is $a+bi+cj+dk\mapsto a-bi-cj-dk$. 

Over $\C$, finite simple $*$-algebras are classified as follows. To start with, 
recall that a finite simple algebra $S$ over $\C$ is isomorphic to $M(n, \C)$, 
because the only 
finite dimensional division algebra over an algebraically closed field is the 
field itself.
So 
without loss of generality we may assume $S$ is $M(n, \C)$. Then any involution 
$*$ on $S$ is in one of the following forms. 
\begin{description}
\item[Orthogonal type] For $X\in M(n, \C)$, 
$X^*=A^{-1}X^tA$, $A\in \GL(n, \C)$, $A=A^t$.
\item[Symplectic type] For $X\in M(n, \C)$, 
$X^*=A^{-1}X^tA$, $A\in \GL(n, \C)$, $A=-A^t$.
\end{description}



\paragraph{Adjoint algebras of $\epsilon$-symmetric matrix tuples.} 
We first present the formal definition.
\begin{definition}\label{def:adjoint}
Let $\F$ be a field and fix $\epsilon\in\{1,-1\}$. For $\vecB=(B_1, \dots, B_m)\in 
S^\epsilon(n, \F)^m$, the \emph{adjoint algebra} of $\vecB$, denoted 
as $\Adj(\vecB)$, is $\{(A, D)\in M(n, \F)^{op} \oplus M(n, \F) | \forall 
i\in[m], A^tB_i=B_iD\}$. $\Adj(\vecB)$ is a $*$-algebra over $\F$ with $(A, 
D)^*=(D, A)$.
\end{definition}
Note that it is a subalgebra of $M(n, \F)^{op}\oplus M(n, \F)$, $\F$ embeds in as 
$(\alpha I_n, \alpha I_n)$ for $\alpha\in \F$, and $*$ fixes $\F$. If $\vecB$ is 
non-degenerate then the projection of $\Adj(\vecB)$ to either $M(n, \F)^{op}$ or 
$M(n, \F)$ is faithful. Therefore, \emph{in the non-degenerate case}, we 
can 
identify $(\Adj(\vecB), *)$ as a 
subalgebra of $M(n, \F)$ consisting of $\{D\in M(n, \F) \mid \exists A\in M(n, \F) 
\text{ s.t. } \forall i\in[m], A^tB_i=B_iD\}$, and for $D\in \Adj(\vecB)$, $D^*$ 
is just the (unique) solution of $\forall i\in[m], A^tB_i=B_iD$. In particular we 
have $A^t\vecB=\vecB A^*$.

Note that a linear basis of the 
adjoint algebra of a tuple of $\epsilon$-symmetric matrices can be computed 
efficiently by 
solving a system of linear forms. The $*$-map is also easily implemented. 

\section{Proof of Theorem~\ref{thm:isometry}}\label{sec:iso}

\subsection{An outline of the main algorithm for 
Theorem~\ref{thm:isometry}.}\label{subsec:isometry_outline}

Let $\F$ be a field. Recall that we have 
$\vecB=(B_1, \dots, B_m)$ and 
$\vecC=(C_1, \dots, C_m)\in S^\epsilon(n, \F)^m$. The goal is to 
decide if there exists $F\in\GL(n, \F)$ such that $\forall i\in[m]$, $F^tB_iF=C_i$.
The main steps of the algorithm are as follows.  

\begin{enumerate}
\item {\it Reduce to the non-degenerate case.} If $\vecB$ is 
degenerate, that is 
$\cap_{i\in[m]}\ker(B_i)\neq \vzero$, we can 
reduce to the non-degenerate case by restricting to the non-degenerate part. See 
Section~\ref{subsec:I}.
\item {\it Solve the twisted equivalence problem.} In this step we test whether 
$\vecB$ and $\vecC$ are ``twisted equivalent'', that is, whether there 
exist $A, D\in\GL(n, q)$ such that $A^t\vecB=\vecC D$. This problem 
can be solved efficiently by reducing to the module isomorphism problem. See 
Section~\ref{subsec:II}.
\item {\it Reduce to decomposing a symmetric element in a $*$-algebra.} At the 
beginning of this step we know that $\vecB$ and $\vecC$ are twisted equivalent 
under some 
$A, D\in\GL(n, q)$. Note that if $D=A^{-1}$ then we are done. If not, the hope is 
to transform $A$ and $D$ appropriately to get an invertible matrix $F$ such that 
$\vecB$ and $\vecC$ 
are twisted equivalent under $F$ and $F^{-1}$, if such an $F$ exists. Let 
$E=A^{-1}D^{-1}$. Since $\vecC$ is non-degenerate, the adjoint 
algebra of $\vecC$ can be defined alternatively as a 
subalgebra of $M(n, \F)$, $\cA=\Adj(\vecC):=\{D\in M(n, \F) \mid \exists A\in M(n, 
\F) 
\text{ s.t. } \forall i\in[m], A^tC_i=C_iD\}$. The involution $*$ sends $D\in 
\Adj(\vecC)$ to $D^*$, which is the (unique) solution of $\forall i\in[m], 
A^tC_i=C_iD$. 
It can be verified that $E\in \cA$, and $E^*=E$. The 
important observation then is that, there exists such $F$ if and only if there 
exists $X\in \cA$ such that $E=X^*X$. See Section~\ref{subsec:III}.
%
\item {\it Solve the $*$-symmetric decomposition problem.} This is the main 
technical piece of this algorithm. This step relies on certain results about the 
structure of 
$*$-algebras, which has been summarized in Section~\ref{sec:prel}. The basic idea 
is to utilize 
the 
algebra 
structure of $\cA$, to reduce to the semisimple case, and then further to the 
simple case. To deal with the simple case turns out to be 
exactly the isometry problem for a \emph{single} (symmetric, skew-symmetric, or 
Hermitian\dots) form, which can be solved using existing algorithms. We now 
outline the main steps.
\begin{enumerate}
\item[4.a.] {\it Compute the algebra structure of $\cA$.} We start with computing 
the 
algebra structure of $\cA$, including the Jacobson radical $\Rad(\cA)$, the 
decomposition of 
the semisimple quotient into simple summands, and for each simple summand, an 
explicit isomorphism with a matrix ring over a division algebra. This can be  
achieved by resorting to known algorithms by R\'onyai \cite{Ron90} and Eberly 
\cite{Eber91a,Ebe91}. This step is the main bottleneck to extend this algorithm to 
number fields (without going to extension fields). See Section~\ref{subsubsec:I}. 
\item[4.b.] {\it Recognize the $*$-algebra structure.} We then take into account 
the 
$*$-algebra structure. The involution $*$ preserves the Jacobson radical, 
so it induces an involution on the semisimple quotient, denoted again by $*$. For 
a particular 
summand $S$ of the semisimple quotient, $*$ either switches $S$ with another 
summand, or preserves it. In the 
the latter case, by the 
structure theory of $*$-algebras in the simple case, $*$ has to be in a 
particular 
form, and this form can be computed explicitly by resorting to the module 
isomorphism problem. See Section~\ref{subsubsec:II}.
\item[4.c.] {\it Reduce to the semisimple case.} In this step, we show that any 
solution to the $*$-symmetric decomposition problem for $\cA/\Rad(\cA)$ and 
$E+\Rad(\cA)$ can be 
lifted efficiently to a solution to the $*$-symmetric decomposition problem for 
$\cA$ and 
$E$. This procedure crucially relies on that we work with fields of characteristic 
not 
$2$, and is the main bottleneck to extend this algorithm to fields of 
characteristic $2$. This means that we can reduce to work with semisimple 
$*$-algebra $\cA$ in 
the following. 
See Section~\ref{subsubsec:III}.
\item[4.d.] {\it Reduce to the $*$-simple and simple case.} In this step, we want 
to  
tackle the $*$-symmetric decomposition problem for a semisimple $*$-algebra $\cA$. 
Recall 
that a decomposition of $\cA$ as a sum of simple summands has been computed in 
Step (4.a). We present a reduction to the same problem for those simple summands 
that are preserved by 
$*$. This means that we can reduce to work with a simple $*$-algebra $\cA$. See 
Section~\ref{subsubsec:IV}.
\item[4.e.] {\it Tackle the simple case by reducing to the isometry problem for a 
single 
form.} In this step, we want to solve the $*$-symmetric decomposition problem for 
a 
simple $*$-algebra $\cA$. Recall that an explicit isomorphism of $\cA$ with a 
matrix ring over a division algebra has been computed in Step (4.a), and a 
particular form of $*$ on $\cA$ has been computed in Step (4.b). By these two 
pieces of information, we can reduce the $*$-symmetric decomposition problem for 
$\cA$ to 
the isometry problem for a \emph{single} classical (symmetric, skew-symmetric, 
Hermitian\dots) form. See 
Section~\ref{subsubsec:V}.
\item[4.f.] {\it Solve the isometry problem for a single form.} To solve the 
isometry 
problem for a single classical form is a classical algorithmic problem. One 
approach is to transform a given form into the standard form, by first block 
diagonalizing it, and then bringing the diagonal blocks to basic ones. Do this 
for both 
forms, compare whether the respective standard forms are the same, and if so, 
recover the isometry from the changes of bases in the standardizing procedures. 
See 
Section~\ref{subsubsec:VI}.
\end{enumerate}
\end{enumerate}

From Step (4.f) above, we may view the whole procedure as a reduction from 
isometry testing of an $\epsilon$-symmetric matrix tuple to isometry testing of 
classical forms. Over $\R$, these classical forms are exactly those ones 
that define the classical groups in the sense of Weyl \cite{Wey97} (see 
Section~\ref{sec:prel}). In particular, in principle all possible classical forms
-- symmetric, skew-symmetric, Hermitian, 
skew-Hermitian over $\R$, $\C$, and the quaternion algebra $\qH$ -- can arise, 
even when we deal with only a symmetric matrix tuple. It will be interesting 
to implement our algorithm and examine whether every classical form type indeed 
arises.

There is a tricky issue if we want to output an isometry over 
$\R$ and $\C$ as described in Theorem~\ref{thm:isometry} (2) and (3). Over 
$\R$ and $\C$, the simple summands of a semisimple algebra may be 
defined over different extension fields, and one needs to be careful not to mix 
these fields arbitrarily as that may lead to an extension field of exponential 
degree. To 
overcome this problem we need an alternative solution to the $*$-symmetric 
decomposition 
problem as described in Section~\ref{subsec:alt},
based on 
$*$-invariant Wedderburn-Malcev complements of the Jacobson ideal of a $*$-algebra 
\cite{Taf57}.

In the following subsections, from Section~\ref{subsec:I} to~\ref{subsec:decomp}, 
we give the detailed procedure, which solves completely the case of 
$\F_q$, 
as well as the decision version of the isometry problem for $\R$ and $\C$. The 
main algorithm fails to construct an explicit isometry as described in 
Theorem~\ref{thm:isometry} (2) and (3). We remedy this by providing an alternative 
algorithm in Section~\ref{subsec:alt}, which replaces some steps of the main 
algorithm.

\subsection{Main algorithm I: reduce to the non-degenerate case.}\label{subsec:I}

This step works over any field. The procedure is standard but we 
give details here for completeness. 

Recall that $\vecB\in S^\epsilon(n, \F)^m$, as an $\epsilon$-symmetric matrix 
tuple, is non-degenerate if $\ker(\vecB)=\vzero$
(Section~\ref{sec:prel}). Now suppose we are given $\vecB\in S^\epsilon(n, \F)^m$, 
and let 
$d=\dim(\ker(\vecB))$. Form a change of 
basis matrix $S=[v_1, \dots, v_n]$, $v_i\in \F^n$, such that $\{v_{n-d+1}, \dots, 
v_n\}$ is a basis of $\ker(\vecB)$, and $\langle v_1, \dots, v_{n-d}\rangle$ is a 
complement subspace of $\ker(\vecB)$. Then for every $i\in[m]$, 
$S^tB_iS=\begin{bmatrix}
B_i' & 0 \\
0 & 0
\end{bmatrix}$
where $B_i'\in S^\epsilon(n-d, \F)$. We call $\vecB'=(B_1', \dots, B_m')$ a 
non-degenerate tuple extracted from $\vecB$. It is easy to show the 
following.
\begin{proposition}
Given $\vecB, \vecC\in S^\epsilon(n, \F)^m$, let $\vecB'\in S^\epsilon(\ell_1, 
\F)^m$ (resp. $\vecC'\in S^\epsilon(\ell_2, \F)^m$) be a non-degenerate tuple 
extracted from $\vecB$ (resp. $\vecC$). Then $\vecB\sim \vecC$ if and only if 
$\ell_1=\ell_2$, and 
$\vecB'\sim \vecC'$.
\end{proposition}

Since extracting a non-degenerate tuple from $\vecB$ involves only standard linear 
algebraic computations, this step can be performed in deterministic polynomial 
time. So in the following we can assume that $\vecB$ and $\vecC$ are both 
non-degenerate. 
%

\subsection{Main algorithm II: solve the twisted equivalence 
problem.}\label{subsec:II}

This step works over any field. $\vecB, \vecC\in M(n, \F)^m$ are twisted
equivalent, if there exist $A, D\in \GL(n, \F)$ such that $A^t\vecB=\vecC D$. This 
differs from the usual equivalence as in Definition~\ref{def:three} due to the 
transpose of $A$. But any solution $(A, D)$ to the equivalence problem clearly 
gives a solution to the twisted equivalence problem by $(A^t, D)$. The reason to 
introduce the twisted equivalence is because we want to be closer to the isometry 
concept. We now show how to test whether $\vecB$ and $\vecC$ are equivalent, by
a reduction to the module isomorphism problem.

\begin{proposition}
Given $\vecB, \vecC\in M(n, \F)^m$, there exists a deterministic algorithm that 
decides whether $\vecB$ and $\vecC$ are equivalent (and therefore twisted 
equivalent). The algorithm uses 
polynomially many arithmetic operations. Over number fields the bit 
complexity of the 
algorithm is also polynomial. 
\end{proposition}
\begin{proof}
From $\vecB=(B_1, \dots, B_m)$, construct a tuple 
of matrices 
$\vecB'=(B_0', B_1', \dots, B_m')$, where $B_i'\in M(2n, \F)$, as 
follows. Every $B_i$ is viewed as a $2\times 2$ block matrix with each block of 
size 
$n\times n$. $B_0'=\left[
\begin{array}{cc}
I_n & 0 \\
0 & 0
\end{array}
\right]$, and for $i\in[m]$, $B_i'=\left[
\begin{array}{cc}
0 & B_i \\
0 & 0
\end{array}
\right]$. Similarly construct $\vecC'$. 

We claim 
that there exist $A, D\in\GL(n, \F)$ satisfying 
$A\vecB=\vecC D$ if and only if there exists an invertible $E\in\GL(2n, \F)$ 
satisfying $E\vecB'=\vecC'E$. For the if direction, let $E=\begin{bmatrix}A & G 
\\ H 
& D\end{bmatrix}$. By $EB_0'=C_0'E$, we have $G=H=0$. Therefore, as $E\in\GL(2n, 
\F)$, $A, 
D\in\GL(n, \F)$. Furthermore, for 
$i\in[m]$, by $\begin{bmatrix}A & 0 \\ 0 
& D\end{bmatrix}\left[
\begin{array}{cc}
0 & B_i \\
0 & 0
\end{array}
\right]=\left[
\begin{array}{cc}
0 & C_i \\
0 & 0
\end{array}
\right]\begin{bmatrix}A & 0 \\ 0 
& D\end{bmatrix}$, we see that $AB_i=C_iD$. For the only if direction, if 
$AB_i=C_iD$ for all $i\in[m]$, then it is easy to see that $E=\begin{bmatrix}A 
& 0 \\ 0 
& D\end{bmatrix}$ satisfies that $E\vecB'=\vecC'E$. 

Therefore, the above construction gives an efficient reduction from the 
equivalence problem for $\vecB$ and $\vecC$ to the conjugacy problem for $\vecB'$ 
and $\vecC'$. We can then call the procedure in Theorem~\ref{thm:mi} to conclude. 
%
%
\end{proof}

Note that if $\vecB\sim \vecC$ then $\vecB$ and $\vecC$ are indeed twisted 
equivalent. 
In other words, if $\vecB$ and $\vecC$ are not twisted equivalent we conclude that 
they 
are not 
isometric either. Therefore, in the following we assume that we have computed $A, 
D\in \GL(n, \F)$ such that $A^t\vecB = \vecC D$.

\subsection{Main algorithm III: reduce to decomposing a $*$-symmetric element in a 
$*$-algebra.}\label{subsec:III}

This step works over any field. From previous steps, for the non-degenerate 
$\vecB, 
\vecC\in S^\epsilon(n, \F)$, we have computed $A, D\in\GL(n, \F)$ such that 
$A^t\vecB=\vecC D$. 

Let $\cA=\Adj(\vecC)$, with the natural 
involution $*$. Since $\vecC$ is non-degenerate, $\cA$ can be embedded as a 
subalgebra of $M(n, \F)$ (see Section~\ref{sec:prel}.)
Let $E=A^{-1}D^{-1}$. Note that $E$ is invertible. 

\begin{claim}
Let $E$ and $\cA$ be as above. $E$ is a $*$-symmetric element in $\cA$.
\end{claim}
\begin{proof}
Observe that $A^t\vecB=\vecC D\iff \vecB D^{-1}=A^{-t}\vecC\iff 
D^{-t}\vecB^t=\vecC^tA^{-1}\iff 
D^{-t}\vecB=\vecC A^{-1}$, where the last $\iff$ uses that $\vecB$ and $\vecC$ are 
from $S^\epsilon(n, \F)$.
Therefore $(A^{-1}D^{-1})^t\vecC=D^{-t}A^{-t}\vecC=D^{-t}\vecB D^{-1}=\vecC 
A^{-1}D^{-1}$.
\end{proof}

The following proposition is a conceptually crucial observation for the algorithm. 

\begin{proposition}\label{prop:decomposition}
Let $\vecB$, $\vecC$, $\cA$, and $E$ be as above. Then $\vecB\sim \vecC$ if and 
only if there exists $X\in 
\cA$ such that $X^*X=E$.
\end{proposition}
\begin{proof}
For the if direction, by $X^*X=A^{-1}D^{-1}$, we have $AX^*=D^{-1}X^{-1}$. Also 
observe that $D^{-t}\vecB=\vecC A^{-1}$, and $(X^*)^t\vecC=\vecC X\iff \vecC 
X^*=X^t\vecC\iff X^{-t}\vecC=\vecC (X^*)^{-1}$. So 
$(D^{-1}X^{-1})^t\vecB=X^{-t}D^{-t}\vecB=X^{-t}\vecC A^{-1}=\vecC 
(X^*)^{-1}A^{-1}$, which gives $(D^{-1}X^{-1})^t\vecB(AX^*)=\vecC$. Now recall 
that $AX^*=D^{-1}X^{-1}$, so $D^{-1}X^{-1}$ is the desired isometry. 

For the only if direction, suppose $Z^t\vecB Z=\vecC$. Setting $X=Z^{-1}D^{-1}$ 
and $Y=A^{-1}Z$, we have $AY=D^{-1}X^{-1}=Z$. So $\vecC=Z^t\vecB Z=Y^tA^t\vecB 
D^{-1}X^{-1}=Y^t\vecC X^{-1}$, which gives $Y=X^*$. By $YX=A^{-1}D^{-1}$, 
$X^*X=A^{-1}D^{-1}$ follows.
\end{proof}

Proposition~\ref{prop:decomposition} then leads to the following question. 
\begin{problem}[$*$-symmetric decomposition problem]
Let $\cA$ be a matrix algebra in $M(n, \F)$ with an involution $*$, and $E\in \cA$ 
be an invertible $*$-symmetric element. Compute $X\in \cA$ such that $X^*X=E$, if 
there exists 
such an element. 
\end{problem}

\subsection{Main algorithm IV: solve the $*$-symmetric decomposition 
problem.}\label{subsec:decomp}
This is the main 
technical piece of this algorithm. The strategy is to utilize the algebra 
structure of $\cA$, and reduce the problem to the case when $\cA$ is a simple 
algebra. When $\cA$ is simple and can be explicitly represented as a full matrix 
ring 
over division algebras, the problem turns out to be equivalent to solving the 
isometry problem for a single classical (symmetric, skew-symmetric, 
Hermitian\dots) form, which then can be solved using existing algorithms. 

\subsubsection{Decomposition algorithm I: compute the algebra 
structure.}\label{subsubsec:I}
By 
resorting to known results, this 
step works over finite fields \cite{Ron90,Iva00,EG00}, the 
real field, and the complex field \cite{FR85,Eber91a,Ebe91}. We now cite these 
results as follows.

\begin{theorem}[\cite{Ron90}; see also \cite{Iva00,EG00}]\label{thm:algebra_finite}
Suppose we are given a linear basis of an algebra $\cA$ in $M(n, \F_q)$. There is 
a Las Vegas algorithm that computes 
\begin{enumerate}
\item a linear basis of the 
Jacobson radical $\Rad(\cA)$, and
\item an 
epimorphism $\pi: \cA\to M(n_1, \F_{q_1})\oplus \dots \oplus M(n_k, \F_{q_k})$ 
with 
kernel 
$\Rad(\cA)$, and $\F_{q_i}$ an extension field of $\F_q$. $\F_{q_i}$ is 
specified by 
a linear basis over $\F_q$.
\end{enumerate}
The algorithm runs in time 
$\poly(n, \log q)$, and can be derandomized at the price of running in time 
$\poly(n, \log q, p)$ where $p=\fdchar(\F_q)$. 

Furthermore, there are efficient 
deterministic algorithms that 
\begin{enumerate}
\item[i.] given $a\in \cA$, compute $\pi(a)$, and
\item[ii.] given 
$b\in M(n_1, \F_{q_1})\oplus \dots \oplus M(n_k, \F_{q_k})$, compute $a\in \cA$ 
such that 
$\pi(a)=b$.
\end{enumerate}
\end{theorem}

\begin{theorem}[\cite{FR85,Eber91a,Ebe91,Ron94}]\label{thm:algebra_R_and_C}
Let $\E$ be a number field, and suppose we are given a linear basis of an algebra 
$\cA$ in $M(n, \E)$. Then there exists a deterministic polynomial-time algorithm 
that computes 
\begin{enumerate}
\item a linear basis of the Jacobson radical $\Rad(\cA)$ 
over $\E$, and
\item \begin{description}
\item[Over $\R$:] \begin{enumerate}
\item the number $k$ of simple components of $\cA\otimes_\E\R$, 
\item specifications of extension fields $\E\subseteq \E_1, \dots, \E_k\subseteq 
\R$, such that each $\E_i$ is of degree at most $\binom{\dim_\E \cA}{2}$ over 
$\E$, 
\item bases of simple algebras $B_1\subseteq A\otimes_\E \E_1$, \dots, 
$B_k\subseteq A\otimes_\E \E_k$, such that $B_i\otimes_{\E_i}\R$, $i\in[k]$, are 
all the simple components of $\cA\otimes_\E\R$, and 
\item for each $i\in[k]$, an extension field $\K_i\subseteq \R$ over $\E_i$ with 
extension 
degree at most $\dim_{\E_i}B_i$, the linear basis of a division algebra 
$D_i\subseteq 
B_i\otimes_{\E_i}\K_i$ over $\K_i$, and the linear 
basis of a subalgebra 
$M_i\subseteq B_i\otimes_{\E_i}\K_i$ over 
$\K_i$, such that $M_i\cong M(n_i, \K_i)$, and $B_i\otimes_{\E_i}\K_i\cong 
M_i\otimes_{\K_i}D_i\cong 
M(n_i, D_i)$. $\dim_{\K_i}D_i$ can be $1$, $2$, or $4$, and when 
$\dim_{\K_i}D_i=4$, $D_i$ is non-commutative. 
\end{enumerate}
\item[Over $\C$:] \begin{enumerate}
\item the number $k$ of simple components of $\cA\otimes_\E\C$, 
\item specifications of extension fields $\E\subseteq \E_1, \dots, \E_k$, such 
that each $\E_i$ is of degree at most $\dim_\E\cA$ over $\E$,
\item bases of simple algebras $B_1\subseteq A\otimes_\E \E_1$, \dots, 
$B_k\subseteq A\otimes_\E \E_k$, such that $B_i\otimes_{\E_i}\C$, $i\in[k]$, are 
all the simple components of $\cA\otimes_\E\C$, and
\item for each $i\in[k]$, an extension field $\K_i$ over $\E_i$ with extension 
degree at most $\sqrt{\dim_{\E_i}B_i}$, the linear basis of a subalgebra 
$M_i\subseteq B_i\otimes_{\E_i}\K_i$ over $\K_i$, such that $M_i\cong M(n_i, 
\K_i)$.
\end{enumerate}
\end{description}
\end{enumerate}
%

\end{theorem}

\begin{remark}\label{rem:extension_degree}
\begin{enumerate}
\item Comparing Theorem~\ref{thm:algebra_finite} and 
Theorem~\ref{thm:algebra_R_and_C}, 
we see that a statement corresponding to Theorem~\ref{thm:algebra_finite} (ii) 
was missing in Theorem~\ref{thm:algebra_R_and_C}. This is because a preimage of 
$b\in M(n_1, 
D_1)\oplus 
\dots\oplus M(n_k, D_k)$ may live in $\cA\otimes_\E \K$ for some field 
$\K$ with an exponential extension degree over $\E$. This suggests that 
representing the 
isometry in the settings of $\R$ and $\C$ as a \emph{single} matrix would be 
inefficient. 
\item The randomized version of Theorem~\ref{thm:algebra_R_and_C} is shown by 
Eberly in \cite{Eber91a,Ebe91}, and is subsequently derandomized by R\'onyai in 
\cite{Ron94}. To completely derandomize Theorem~\ref{thm:algebra_finite} is a 
difficult problem as this relies on algorithms for polynomial factorization over 
finite fields. 
\end{enumerate}
\end{remark}

\subsubsection{Decomposition algorithm II: recognize the $*$-algebra 
structure.}\label{subsubsec:II}
This step works over $\F_q$ 
with $q$ odd, $\R$, and $\C$. It may be possible 
to handle fields of even characteristics, but we leave it for further study. 
The case of finite fields of odd characteristics has been settled by Brooksbank 
and Wilson in \cite{BW12}. Here we provide a unified and somewhat 
simpler
treatment over those fields just mentioned. 

To start with, recall that from previous steps we have computed the algebra 
structure of $\cA\subseteq M(n, \F)$, including a linear basis of $\Rad(\cA)$ and 
an 
epimorphism $\pi: \cA\to S_1\oplus \dots\oplus S_k$ where $S_i$ is a simple 
algebra over the designated field (after some scalar extension when over $\R$ or 
$\C$). We have also computed explicit isomorphisms between $S_i$ and matrix rings 
over division rings. Since $\Rad(\cA)$ is a $*$-ideal, the involution $*$ induces 
an 
involution, 
which we denote again by $*$, on $\pi(\cA)$. Then for each $S_i$, either 
$S_i^*=S_i$, or $S_i^*=S_j$ for some $j\neq i$. The 
goal is that, in the former case, we want to express the involution $*$ explicitly 
in the forms presented in Section~\ref{sec:prel}.

\begin{proposition}\label{prop:star_recog_field_case}
Let $\E/\F$ be a field extension specified by a linear basis over $\F$. Given an 
involution 
$*$ of $M(n, \E)$ as an $\F$-algebra, there exists a 
deterministic polynomial-time algorithm that (1) decides whether $*$ induces a 
quadratic field involution of $\E$ over a subfield $\E'$, and (2) computes 
$A\in\GL(n, \E)$ such that for every $X\in M(n, \E)$, $X^*=A^{-1}X'^tA$, where 
$X'$ is either $X$ (when $*$ fixes $\E$) or $\overline{X}$ (when $*$ induces a 
quadratic field involution).
\end{proposition}
\begin{proof}
For (1), we apply $*$ to every basis element $b$ in the linear basis of $\E$ over 
$\F$. If $*$ changes none of them, then $\E$ is also invariant under $*$. If $*$ 
changes some of them, the sums $b+b^*$ 
linearly span a subfield $\E'$ such 
that $\E/\E'$ is a quadratic field extension, and $*$ induces the quadratic field 
involution. For (2), for any $X\in M(n, \E)$ let $X'$ be as defined in the 
statement. We take a linear basis $\{B_1, \dots, B_{n^2}\}$ of $M(n, 
\E)$ (the standard basis will do), and set up $YB_i^*=B_i'^t Y$, for $i\in[n^2]$, 
and $Y$ is an $n\times n$ variable matrix. By \cite[Chap. X.4, Theorem 11]{Alb39}, 
there must exist some $A\in\GL(n, \E)$ as a valid solution to $Y$ in the 
above 
equations. 
From the algorithmic 
viewpoint, this is an instance of the module isomorphism problem, and we can apply 
the procedure in Theorem~\ref{thm:mi} to conclude. 
\end{proof}

Note that Proposition~\ref{prop:star_recog_field_case} covers all simple types 
over $\F_q$ with $q$ odd and $\C$, as well as those simple types over $\R$ except 
the two 
quaternion types. We now handle the two quaternion types in the real field 
setting. 
\begin{proposition}
Let $\qH$ be given by a linear basis over $\R$. Given an involution $*$ of $M(n, 
\qH)$ as an $\R$-algebra, there exists a 
deterministic polynomial-time algorithm that computes $A\in\GL(n, \qH)$ such that 
for every $X\in M(n, \qH)$, $X^*=A^{-1}\overline{X}^tA$.
\end{proposition}
\begin{proof}
Let $f:\qH\to M(4, \R)$ be the regular representation of $\qH$ on $\R^4$. Let 
$\{C_1', C_2', 
C_3', C_4'\}$ be a linear basis of the centralizing algebra of $f(\qH)$ in $M(4, 
\R)$, which is isomorphic to $\qH^{op}$. Now 
think of matrices in $M(4n, \R)$ as $n\times n$ block matrices with each block of 
size $4\times 4$. For $i\in[4]$, let $C_i\in M(4n, \R)$ be the diagonal block 
matrix, with all diagonal blocks being $C_i'$. $f$ naturally embeds $M(n, \qH)$ to 
$M(4n, \R)$. By the double centralizer theorem, the centralizing algebra of 
$C_i$'s is $f(M(n, \qH))$. 

The above reasoning suggests the following construction. Take a basis $\{B_1, 
\dots, B_{n^2}\}$ of $M(n, \qH)$, and let $B_i'=\overline{B_i}^t$. Set up 
$Yf(B_i)=f(B_i')Y$, $i\in[n^2]$, $YC_j=C_jY$, $j\in[4]$, where $Y$ is a $4n\times 
4n$ variable matrix. By $YC_j=C_jY$, any valid solution to $Y$ lies in $f(M(n, 
\qH))$. By an analogous argument as in the proof of 
Proposition~\ref{prop:star_recog_field_case}, there must exist an invertible 
$A\in\GL(4n, \R)$ as a valid solution to $Y$, and can be solved as 
as an instance of the module isomorphism problem by Theorem~\ref{thm:mi}.
Finally, after getting such an $A\in \GL(4n, \R)$, it is 
straightforward 
to compute the preimage of $A$ in $M(n, \qH)$, concluding the proof. 
\end{proof}


\subsubsection{Decomposition algorithm III: reduce to the semisimple 
case.}\label{subsubsec:III}

This 
step works over fields of 
characteristic $\neq 2$, and is the main bottleneck for handling fields of 
characteristic $2$. 

\begin{proposition}\label{prop:III}
Let $\cA$ be a $*$-algebra over $\F$, $\fdchar(\F)\neq 2$. Let $E\in \cA$ be an 
invertible
$*$-symmetric element, and suppose there exists $Y\in \cA/\Rad(\cA)$, such that 
$Y^*Y+\Rad(\cA)=E+\Rad(\cA)$. Then there exists $X\in \cA$ such that $X^*X=E$, and 
there 
exists a deterministic polynomial-time algorithm that outputs such an $X$.
\end{proposition}
\begin{proof}
To recover $X\in \cA$ such that $X^*X=E$, consider the 
following situation: suppose we have a $*$-ideal $J$ of $\cA$ with $J^2=0$, and an 
invertible $E\in \cA$ with $E^*=E$. Given $Y$ such that $Y^*Y+J=E+J$, the goal is 
to find $Z\in J$ such that $(Y+Z)^*(Y+Z)=E$. Expanding to $Y^*Y+Y^*Z+YZ^*+Z^*Z=E$, 
by $Z^*Z=0$ we need to satisfy $Y^*Z+Z^*Y=E-Y^*Y$. Note that $E-Y^*Y$ is 
$*$-symmetric. So setting $U=\frac{1}{2}(E-Y^*Y)$, $Z=Y^{-*}U$ is the desired, and 
$X=Y+Z$ satisfies $X^*X=E$. 

We now apply the above procedure to the setting of the proposition. For $i\in \N$, 
let $J_i=\Rad(\cA)^{2^i}$, so that $J_{i+1}=J_{i}^2$. Since the Jacobson radical 
$\Rad(\cA)$ is nilpotent, we 
know that for some $k\leq \lceil \log n\rceil$, $J_k=0$. Given $Y_i\in \cA/J_i$ 
satisfying $Y_i^*Y_i+J_i=E+J_i$, consider $\cA/J_{i+1}$, in which $J_i/J_{i+1}$ 
satisfies the assumption on $J$ in the last paragraph. We can then utilize the 
procedure there to get $Y_{i+1}\in \cA/J_{i+1}$, such that 
$Y_{i+1}^*Y_{i+1}+J_{i+1}=E+J_{i+1}$. Let the given $Y$ satisfying 
$Y^*Y+\Rad(\cA)=E+\Rad(\cA)$ be $Y_0$, and perform the above procedure iteratively 
for 
at most $k\leq \lceil \log n\rceil$ times. We then obtain the desired $X\in \cA$ 
such that  
$X^*X=E$.
\end{proof}

\subsubsection{Decomposition algorithm IV: reduce to the $*$-simple and simple 
case.}\label{subsubsec:IV} This step works for any 
field. Suppose we have a semi-simple algebra $\cA$ decomposed into a 
direct sum of simple summands $S_1\oplus \dots\oplus S_k$, and let $*$ be an 
involution on $\cA$. Without loss of generality, we can assume that there exists 
$j\leq \lfloor k/2\rfloor$, such that $*$ exchanges $S_{2i-1}$ and $S_{2i}$ for 
$i\in[j]$, and stabilizes $S_i$ for $i>2j$. 
Let $E\in \cA$ be an invertible $*$-symmetric 
element, and let $E_i$ be the projection of $E$ to $S_i$. Recall that our goal is 
to find $X\in \cA$ such that $X^*X=E$, if such an $X$ exists. 

\begin{proposition}\label{prop:reduce_to_simple}
Let $\cA$, $S_i$, $E$, and $E_i$ be as above. There exists $X\in \cA$ such that 
$X^*X=E$, if and only if for every $i>2j$, there exists $X_i\in S_i$ such that 
$X_i^*X_i=E_i$.
\end{proposition}
\begin{proof}
For the if direction, we claim that $X=E_1\oplus I\oplus E_3\oplus I\oplus \dots 
\oplus E_{2j-1}\oplus I\oplus X_{2j+1}\oplus \dots\oplus X_k$ is a solution, where 
$I$ denotes the identity element in the respective summand. To see this, let us 
suppose $*$ exchanges $S_1$ and $S_2$. Then by $(E_1, E_2)^*=(E_1, E_2)$, we have 
$(E_1, I)^*=(I, E_2)$. So $X^*X=(I\oplus E_2\oplus I\oplus E_4\oplus \dots 
\oplus I\oplus E_{2j}\oplus X_{2j+1}^*\oplus \dots\oplus X_k^*)(E_1\oplus I\oplus 
E_3\oplus I\oplus \dots 
\oplus E_{2j-1}\oplus I\oplus X_{2j+1}\oplus \dots\oplus X_k)=E$.

For the only if direction, suppose $X=X_1\oplus X_2\oplus \dots \oplus 
X_{2j-1}\oplus X_{2j}\oplus X_{2j+1}\oplus\dots\oplus X_k$ satisfies $X^*X=E$. 
Then it is straightforward to verify that for $i>2j$, $X_i^*X_i=E_i$.
\end{proof}

\subsubsection{Decomposition algorithm V: the simple case by reducing to the 
isometry 
problem for a single form.}\label{subsubsec:V} 
This step 
works over any field. From previous steps, we now have (1) $M(n, D)$ where $D$ is 
a 
field or a division algebra, (2) an involution $*$ on $M(n, D)$, which induces an 
involution $\overline{\cdot}:D\to D$ (possibly identity), such that 
$X^*=A^{-1}\overline{X}^t A$ and $\overline{A}^t=\epsilon A$ for some 
$\epsilon\in\{1, -1\}$, and (3) an 
invertible $*$-symmetric element $E$. 
%

Here is the other conceptually crucial observation.  
\begin{proposition}
Let notation be as above. Let $F=AE$. Then $F$ is a form of the same type as $A$, 
and there exists $X$ such that $X^*X=E$, if and only if $A$ and $F$ are isometric. 
\end{proposition}
\begin{proof}
To see that $F$ is a form of the same type as $A$, we have 
$E=E^*=A^{-1}\overline{E}^tA$ (by the $*$-symmetry of $E$) and 
$\overline{A}^t=\epsilon A$. It follows that $AE=\overline{E}^tA$, from which we 
get
$\overline{AE}^t=\overline{E}^t\overline{A}^t=\epsilon\overline{E}^tA=\epsilon AE$.

For the second statement, we consider the if direction first. If for some 
$Y\in\GL(n, D)$, $Y^tA\overline{Y}=F=AE$, then $A^{-1}Y^tA\overline{Y}=E$. Setting 
$X=\overline{Y}$, we have $A^{-1}\overline{X}^tA X=E$. Noting that 
$A^{-1}\overline{X}^tA=X^*$, we obtain the desired $X^*X=E$. The only if direction 
can be seen easily by inverting the above reasoning. 
\end{proof}

\subsubsection{Decomposition algorithm VI: solve the isometry problem for a 
single 
form.}\label{subsubsec:VI}
To solve the 
isometry problem for a single form over a division ring, we will in fact compute 
the canonical form for such a form. The isometry problem can then be solved by 
comparing the canonical forms. Over $\F_q$ with $q$ odd, a concrete isometry can 
be obtained by using 
the transformations to the canonical forms. To recover a concrete isometry 
(represented in some form) over 
$\R$ or $\C$ requires more technical machinery and we leave it to 
Section~\ref{subsec:alt}. The existence of canonical forms is 
well-known for $\F_q$ with $q$ odd (see e.g. \cite[Chap. 3.4]{Wilson_book}), for 
$\R$ 
(see e.g. \cite{Lew77}), and for $\C$. 

Computing the canonical form 
involves two steps. Let $E\in M(n, D)$ such that $\overline{E}^t=\epsilon E$, 
where $\overline{\cdot}:D\to D$ is an involution, and $\epsilon\in \{1, -1\}$. 

The 
first step is to compute an orthogonal basis for $E$, that is a linear basis of 
$D^n$ $\{e_1, \dots, e_n\}$, such that for every $i\in[n]$, 
$e_i^tE\overline{e_j}\neq 0$ for exactly one $e_j$. This 
is known as the Gram-Schmidt procedure, and an efficient algorithm in this general 
setting has been obtained by Wilson. 
\begin{theorem}[\cite{Wil13Gram}]
Let $E$ be as above. There exists a deterministic polynomial-time algorithm that 
computes an orthogonal basis for $E$.
\end{theorem}

After the first step, by transforming to the orthogonal basis, $E$ can be assumed 
to be a diagonal block matrix, with each block is of size $1$ or $2$. The second 
step is to simplify these diagonal blocks as much as possible. We now need to 
handle each field separately. Recall that $E$ is non-degenerate.

\paragraph{Block diagonal forms over $\F_q$.} We distinguish among the three 
simple 
types over $\F_q$. 
\begin{description}
\item[Orthogonal type] In this case, each block is of size $1$, e.g. $E$ is a 
diagonal matrix. Fix a non-square $\omega$ in $\F_q$, which can be computed 
efficiently, by either using randomness, or in a deterministic way if we assume 
the Generalized Riemann Hypothesis or the characteristic of $\F_q$ is small. We 
can first simplify $E$ as $\diag(1, \dots, 1, \omega, \dots, 
\omega)$, by resorting to square root computations over finite fields. This can 
be done in randomized 
polynomial time by e.g. the Tonelli-Shanks algorithm. A deterministic 
polynomial-time algorithm exists, if we assume the Generalized Riemann hypothesis, 
or the characteristic of the finite field is small. Then, if the number of 
$\omega$'s is larger than $1$, then write $\omega$ as a sum of two squares 
$\alpha^2+\beta^2$, which is always possible over a finite field. Algorithmically, 
this can be done by solving the equation 
$x^2+y^2=\omega$ in deterministic polynomial time by an algorithm of van de 
Woestijne \cite[Theorem A.3]{Woe05}. Given such $\alpha, \beta$,
$\diag(\omega, \omega)$ can be transformed to $\diag(1, 1)$ by $\begin{bmatrix}
\alpha & \beta \\
\beta & -\alpha 
\end{bmatrix}
\begin{bmatrix}
1 & 0 \\
0 & 1 
\end{bmatrix}
\begin{bmatrix}
\alpha & \beta \\
\beta & -\alpha 
\end{bmatrix}^t
=
\begin{bmatrix}
\omega & 0 \\
0 & \omega
\end{bmatrix}.
$ 
Therefore the possible standard forms are $\diag(1, \dots, 
1)$ or $\diag(1, \dots, 1, \omega)$.
\item[Symplectic type] In this case, each block is of size $2$, so we examine one 
block $\begin{bmatrix}
0 & \alpha \\
-\alpha & 0
\end{bmatrix}$. Now by expressing $\alpha$ as a sum of squares, similar trick 
applies to bring it to $\begin{bmatrix}
0 & 1 \\
-1 & 0
\end{bmatrix}$. 
Indeed, the standard form for a 
non-degenerate alternating bilinear form of size $2k\times 2k$ is the block 
diagonal 
matrix, with each block on the diagonal being $\begin{bmatrix}
0 & 1 \\
-1 & 0
\end{bmatrix}$ (see e.g. \cite[Sec. 3.4.4]{Wilson_book}).
\item[Hermitian type] In this case, each block is of size $1$. Let the associated 
field extension be $\F_q/\F_{q'}$ where $q=q'^2$, and suppose 
$\F_q=\F_{q'}(\tau)$, where $\tau$ is a square root of a non-square $\omega\in 
\F_{q'}$. Then for $\alpha=a+b\tau$, 
$\overline{\alpha}=a-b\tau$. For a diagonal 
entry $\alpha\in \F_q$, $\alpha=\overline{\alpha}$, we need to compute $\beta\in 
\F_q$ such that $\beta\overline{\beta}=\alpha$, which always exists. Setting 
$\beta=x+y\tau$, we need to solve the equation 
$\beta\overline{\beta}=x^2-y^2\tau^2=\alpha$. Again 
this can be solved in deterministic polynomial time by \cite[Theorem A.3]{Woe05}. 
Indeed, there always exists an orthonormal basis for a 
non-degenerate Hermitian form, so the standard form is just the identity matrix 
(see e.g. \cite[Sec. 3.4.5]{Wilson_book}).
\end{description}

\paragraph{Block diagonal forms over $\R$.} By \cite{Lew77}, for the symplectic, 
complex 
orthogonal, 
complex symplectic, quaternion orthogonal types, we can always bring a given form 
to the identity matrix or the standard non-degenerate skew-symmetric matrix. For 
other types, we can bring a given form to $\diag(1, 
\dots, 1, 
-1, \dots, -1)$, where the number of $1$'s and the number of $-1$'s is called the 
signature of the canonical form. 

\paragraph{Block diagonal forms over $\C$.} By \cite{Lew77}, for the two types 
here we can always 
bring a given form to 
the identity matrix  or the standard non-degenerate skew-symmetric matrix. 
\begin{remark}
Starting from two $\epsilon$-symmetric matrix tuples $\vecB$ and $\vecC$, suppose 
they are twisted equivalent. We then perform the operations as above, so that for 
each simple component of the semisimple quotient of the $*$-algebra 
$\cA=\Adj(\vecC)$, we have a 
pair of forms. The question of whether $\vecB$ and $\vecC$ are isometric then 
reduces to test whether these pairs of forms are isometric, or in other words, to 
compare
whether they have the same standard form. In particular, over $\C$, since for each 
simple type there exists only one standard form, the twisted equivalence of 
$\vecB$ and $\vecC$ already implies that they are isometric. 
\end{remark}

\subsection{An alternative algorithm for the isometry problem}\label{subsec:alt}

In this section we work over $\R$ and $\C$. We now present an algorithm that, in 
the $\R$ and $\C$ 
settings, can output an 
explicit isometry, which is represented by a product of several matrices, where 
each matrix 
is over an extension field of the number field $\E$ with polynomial extension 
degree, and the 
entries are of polynomial bit sizes. 

We still follow the main algorithm, steps I to III, as described in 
Sections~\ref{subsec:I} 
to~\ref{subsec:III}, to reduce to solving the decomposition problem, 
Problem~\ref{prop:decomposition}, for a $*$-algebra $\cA$ in $M(n, \E)$. Then 
recall that, by the decomposition algorithm steps I to III, as described in 
Sections~\ref{subsubsec:I} to~\ref{subsubsec:III}, we 
can reduce to the semisimple setting and then further to the simple setting. In 
this simple setting, however, we need to work with different extension fields for 
different simple summands, and one cannot mix all those extension fields because 
that would result in an extension field with exponential extension degree (see 
Remark~\ref{rem:extension_degree}). Even within each summand, since we need to 
take square roots to bring the forms into canonical forms, these square roots 
cannot mix arbitrarily because of the same problem. 

\subsubsection{Alternative decomposition algorithm III}
Compare with with Section~\ref{subsubsec:III}.

To tackle these problems, we first devise another reduction to the semisimple 
case, 
based on the existence of $*$-invariant Wedderburn-Malcev complements over fields 
of characteristic $\neq 2$ \cite{Taf57}. The following constructive version of 
Taft's result \cite{Taf57} is by Brooksbank and Wilson \cite{BW12}, in conjunction 
with the algorithm from \cite{GIKR97} that computes a Wedderburn-Malcev complement 
over number fields.

\begin{proposition}[{\cite[Proposition 4.3, Remark 4.2]{BW12}}]
Let $\E$ be a number field, and $\cA\subseteq M(n, \E)$ a $*$-algebra. Then there 
exists a deterministic polynomial-time algorithm that computes a linear basis of 
$\Rad(\cA)$, and a linear basis of a subalgebra $S$, such that 
$\cA=\Rad(\cA)\oplus 
S$ and $S^*=S$.
\end{proposition}
\begin{proof}
The statement on computing $\Rad(\cA)$ is already in 
Theorem~\ref{thm:algebra_R_and_C}. The procedure to compute $S$ is in \cite{BW12}, 
and for completeness we include a sketch. We then note that the bit complexity is 
also polynomially bounded. 

We first resort to Theorem~\ref{thm:algebra_R_and_C} to compute a linear basis of 
$\Rad(\cA)$. We then use the algorithm in \cite[Theorem 3.1]{GIKR97} to compute a 
Wedderburn-Malcev complement $S'$ of $\cA$ in deterministic polynomial time. If 
$\Rad(\cA)=0$ then $\cA$ itself is 
what we want. If $\Rad(\cA)\neq 0$, let $\pi:\cA\to S'$ be the natural 
projection. The involution $*$ induces an involution $\circ$ on $S'$ by sending 
$s\in S'$ to $\pi(s^*)$. Suppose $S'$ is generated by $\{s_1, \dots, 
s_\ell\}$. Let $S''$ be the algebra generated by 
$\{1/2(s_1+s_1^{\circ*}), 
\dots, 1/2(s_\ell+s_\ell^{\circ*})\}$. We then can
reduce to compute a $*$-invariant Wedderburn-Malcev complement in 
$\Rad(\cA)^2\oplus S''$. The number of iterative calls is at most $\lceil\log 
n\rceil$. 

Finally, note that in each iteration the operations are $*$-maps and projections, 
which only increase the bit size by an additive factor of polynomial size. 
Therefore the bit complexity of the above procedure is also polynomial. 
\end{proof}

Based on the Taft decomposition, we can reduce to the semisimple case in a more 
transparent way as follows. 

\begin{proposition}\label{prop:III_alt}
Let $\Rad(\cA)\oplus S$ be a Taft decomposition of a $*$-algebra $\cA\subseteq 
M(n, \E)$. Given a $*$-symmetric element $a'\in\cA$, there is a deterministic 
polynomial-time algorithm 
that computes $u\in 1+\Rad(\cA)$, such that $u^*a'u=a\in S$.
\end{proposition}
\begin{proof}
Let $a'$ be decomposed as $a+r$ with $a\in S$ and $r\in \Rad(\cA)$. We will show 
how to find $t\in \Rad(\cA)$ such that $(1+t)^*a'(1+t)=a+r'$ for $r'\in 
\Rad(\cA)^2$. Then by iterating such a procedure to get $r''\in \Rad(\cA)^4$, 
\dots, we would be done. 

To start with, by the $*$-symmetry of $a'$, $\Rad(\cA)$, and $S$, we have that $a$ 
and $r$ are both $*$-symmetric as well. We expand 
$(1+t)^*a'(1+t)=a+r+t^*a+at+t^*r+rt+t^*at+t^*rt$. Since
$t^*r+rt+t^*at+t^*rt\in \Rad(\cA)^2$, we need $r+t^*a+at=0$. This can be 
achieved by setting  $t=-\frac{1}{2}a^{-1}r$, noting that
$t^*a=t^*a^*=(at)^*$. We then have $(1+t)^*a'(1+t)=a+r'$ for $r'\in \Rad(\cA)^2$. 

To prove that the bit complexity is polynomial, we note that the number of 
iterations is at most $\log n$, and in the $\ell$th iteration we get at most 
$16^\ell$ 
words in the 
alphabet $\{r, a^{-1}\}$, with each word of length at most $5^\ell$. The latter is 
because, if we let $r_\ell$ be the residue in the $\ell$th step, then 
$r_{\ell+1}=t^*r_\ell+r_\ell t+t^*at+t^*r_\ell t=-\frac{3}{4}r_\ell 
a^{-1}r_\ell+\frac{1}{4}r_\ell a^{-1}r_\ell a^{-1}r_\ell$.
\end{proof}

Given $u\in 1+\Rad(\cA)$ such that $u^*a'u=a\in S$, if we can decompose 
$a=x^*x$, then $xu^{-1}$ is a solution for $a'$. The advantage over the procedure 
in Proposition~\ref{prop:III} is the following. If $x$ is represented as a product 
of matrices, each of which is over a different extension field, then the procedure 
in Proposition~\ref{prop:III} may mix these entries over different extension 
fields and cause an extension degree blow-up. On the other hand, the procedure in 
Proposition~\ref{prop:III_alt} takes $x$ and returns $xu^{-1}$, which is still a 
product of matrices, as the output, therefore avoiding the extension degree 
blow-up issue. 

\subsubsection{Alternative decomposition algorithm IV}
Compare with Section~\ref{subsubsec:IV}. 

We now reduce to work with a semisimple $*$-algebra $\cA$ in $M(n, \E)$ and a 
$*$-symmetric element $E$. By 
Theorem~\ref{thm:algebra_R_and_C}, we have extension fields $\E\subseteq\E_i$ and 
simple 
algebras $S_i\subseteq \cA\otimes_\E \E_i$, $i\in[k]$, such that the extension 
degree of $\E_i$ over $\E$ is upper bounded by $\binom{\dim_\E\cA}{2}$ in the real 
case and $\dim_\E\cA$ in 
the complex case. 
We reduce to the simple case by the 
following construction. Without loss of generality, we can assume that there 
exists 
$j\leq \lfloor k/2\rfloor$, such that $*$ exchanges $S_{2i-1}$ and $S_{2i}$ for 
$i\in[j]$, and stabilizes $S_i$ for $i>2j$. Let $E_i$ be the projection of $E$ to 
$S_i$. For $i\leq 2j$, $i$ odd, $(E_i, I)$ is the solution to the decomposition 
problem for $(E_i, E_{i+1})\in S_i\oplus S_{i+1}$. For $i>2j$, suppose $X_i\in 
S_i$ satisfies $X_i^*X_i=E_i$. We then embed $(E_i, I)$ into $\cA\otimes_\E\E_i$, 
and $X_i$ into $\cA\otimes_\E\E_i$, by adding identities in other summands. Let 
$X$ be the product of these matrices. It is easy to see that $X^*X=E$ over some 
extension field $\K$ ($\K$ needs to include all $\E_i$). Note that 
$X$ is then represented by a product of matrices, with each matrix over a 
possibly different extension field. 

\subsubsection{Alternative decomposition algorithm VI}
Compare with Section~\ref{subsubsec:VI}.

We then follow Section~\ref{subsubsec:V} to reduce to the isometry problem for 
a single form. Note that to solve the isometry problem, we need to take square 
roots, which, if not handled well, may lead to extension fields of exponential 
extension degree. Therefore, we also output a product of matrices as an isometry 
between two single forms, keeping those diagonal matrices with square roots on the 
diagonal intact. Specifically, when working with two forms $A$ and $F$ over $\K$, 
the isometry is represented as $T'D'D^{-1}T^{-1}$ where $T$ and $T'$ are the 
orthogonal transformations, and $D'$ and $D$ are diagonal matrices with entries 
being various square roots. We can also represent $D'D^{-1}$ as a single diagonal 
matrix with entries being from an extension field of degree at most $4$.

\section{Proof of Theorem~\ref{thm:sym}}\label{sec:sym}

Recall that in the $\epsilon$-symmetrization problem, we are given a 
matrix tuple $\vecB=(B_1, \dots, B_m)\in M(n, \F)^m$, and need to decide whether 
there exist $A, D\in\GL(n, \F)$ such that $\forall i\in[m]$, $AB_iD$ is 
$\epsilon$-symmetric. 
In Section~\ref{subsec:sym_outline}, we present an algorithm 
when (1) $\F$ is 
large enough, and (2) the Jacobson radical of a matrix algebra can be computed 
efficiently in a deterministic way. Note that (2) holds for 
fields of characteristic $0$ \cite{Dic23}, 
finite fields \cite{Ron90}, as well as many others 
of positive characteristic \cite{CIW}.
 This algorithm follows the strategy for module 
isomorphism problem as used in \cite{CIK97}, 
and relies crucially on 
Lemma~\ref{lem:key}.

We will deal with the remaining cases (a) $|\F|$ 
is large enough but we do not assume the ability to compute the Jacobson radical 
in Section~\ref{app:sub:large_field}, and 
(b) $|\F|$ is small in Section~\ref{app:sub:small_field}. The algorithm for (a) is 
obtained by associating certain projective modules to right ideals,
and adapting the algorithm in Section~\ref{subsec:sym_outline} to work with that 
concept. The algorithm for (b) 
follows the strategy for module isomorphism problem as used in \cite{BL08}, and 
relies crucially on another lemma about $*$-algebra, namely Lemma~\ref{lem:star1}.

To start, note that if 
$\dim(\cap_{i\in[m]}\ker(B_i))+
\dim(\langle \cup_{i\in[m]} \im(B_i)\rangle)\neq n$, then $\vecB$ cannot be 
$\epsilon$-symmetrizable. 
This is because, if $\vecB$ is $\epsilon$-symmetric, then 
$\cap_{i\in[m]}\ker(B_i)$ and $\langle \cup_{i\in[m]} \im(B_i)\rangle$ are 
orthogonal to each other with respect to the standard inner product of 
vectors, so
their dimensions sum up to $n$. Then observe that 
$\dim(\cap_{i\in[m]}\ker(B_i))=\dim(\cap_{i\in[m]}\ker(AB_iD))$, and $\dim(\langle 
\cup_{i\in[m]} \im(B_i)\rangle)=\dim(\langle \cup_{i\in[m]} \im(AB_iD)\rangle)$.
If 
$\dim(\cap_{i\in[m]}\ker(B_i))+\dim(\langle \cup_{i\in[m]} \im(B_i)\rangle)= n$ 
but $\cap_{i\in[m]}\ker(B_i)\neq \vzero$ then we can reduce to the 
$\cap_{i\in[m]}\ker(B_i)=\vzero$ analogously as it is done in Step (1) for the 
isometry problem (Section~\ref{subsec:I}). So in the following we assume 
$\cap_{i\in[m]}\ker(B_i)=\vzero$ and $\langle \cup_{i\in[m]} \im(B_i)\rangle=\F^n$.

\subsection{An algorithm for Theorem~\ref{thm:sym} under certain technical 
conditions}\label{subsec:sym_outline}


%


In this section we present an algorithm for Theorem~\ref{thm:sym} when (1) $\F$ is 
large enough, and (2) the Jacobson radical of a matrix algebra can be computed 
efficiently in a deterministic way.

Recall that, as explained at 
the beginning of Section~\ref{subsubsec:pit}, the $\epsilon$-symmetrization 
problem is equivalent to ask 
whether there exists $E\in\GL(n, \F)$ such that $E\vecB\in S^\epsilon(n, \F)^m$. 
That 
is, whether the matrix space $L^\epsilon(\vecB):=\{Z\in M(n, \F) : \forall 
i\in[m], 
ZB_i=\epsilon B_i^tZ^t\}$ contains a full-rank matrix. 
A linear basis 
$Z_1,\ldots,Z_\ell$ of $L^\epsilon(\vecB)$ can be computed efficiently. 

The remaining part of the algorithm is an iteration during which 
we maintain a matrix $Z\in L^\epsilon(\vecB)$. If $Z$ has full rank we are
done. Otherwise we try all basis elements $Z_i$ and scalars $\lambda$
from a sufficiently large subset $S\subseteq \F$, 
either to  obtain 
a 
matrix 
$Z'=Z+\lambda Z_i$ which is of higher rank 
than $Z$, 
or, if every such $Z'$ is of rank no more than that of 
$Z$, conclude that $Z$ is of the highest rank. We intend to
use the following well known fact. 
Let
$A=\begin{pmatrix}A_{11} & 0 \\ 0 & 0 \end{pmatrix}$ and
$B=\begin{pmatrix}B_{11} & B_{12} \\  B_{21} & B_{22} \end{pmatrix}$
be $(r'+r'')$ by $(r'+r'')$ block matrices, where $A_{11}$ is an $r'$ by $r'$ 
matrix
of rank $r'$ and $B_{22}$ is a nonzero $r''$ by $r''$ matrix.
Let $r=r'+r''$.
Then the matrix $B+\lambda A$ has rank larger than $r'$ for some 
$\lambda$ from a sufficiently large set of scalars. Formally
(see e.g. 
\cite[Lemma 2.2]{IKS10}), 
\begin{lemma}\label{lem:inc_rank}
Let $A, B\in M(r, \F)$ and let 
$S\subseteq \F$ such that $|S| > r$. If 
$B\ker(A)\not\subseteq \im(A)$ then $\rk(\lambda 
A+B)>\rk(A)$ for all but at most 
$r$ $\lambda\in S$.
\end{lemma}
Unfortunately, we are unable to show -- and probably it is not true in general
---
that Lemma~\ref{lem:inc_rank} becomes applicable to $Z$ (as $A$)
and at least one of the basis elements $Z_i$ (as $B$), when we
consider $L^\epsilon(\vecB)$ as it is obviously given to us 
(i.e., a space of $n$ by $n$ matrices). However, there is another
representation
of $L^\epsilon(\vecB)$ as a matrix space in which it provably
does. And this
is the point where $*$-algebras enter the picture.

To see the details, assume that $\vecB=E\vecB'$ where 
$E\in\GL(n, \F)$ and $\vecB'\in 
S^\epsilon(n, \F)^m$. Since $\vecB'$ is non-degenerate, we can identify 
$\Adj(\vecB')\subseteq M(n, \F)^{op}\oplus M(n, \F)$ as a subalgebra of $M(n, \F)$ 
by projecting to the second component (see Section~\ref{sec:prel}). Then 
$L^\epsilon(\vecB')$ is the set 
of 
$*$-symmetric 
elements in $\Adj(\vecB')$. Moreover, it is not difficult to see
that $L^\epsilon(\vecB)=L^\epsilon(\vec B')E^{-1}$. 
Now for $Z\in L^\epsilon(\vecB)$, consider the following composite linear map, 
$Z\mapsto ZE\mapsto
\overline{ZE}\mapsto \ell_{\overline{ZE}}$, where 
$\overline{ZE}=ZE+\Rad(\Adj(\vecB'))$, and
$\ell_{\overline{ZE}}$ is the action of $\overline{ZE}$
on the factor $\Adj(\vecB')/\Rad(\Adj(\vecB'))$ (see also Section~\ref{sec:prel}). 
The following lemma ensures that this gives a representation of 
$L^\epsilon(\vecB)$ to which Lemma~\ref{lem:inc_rank} becomes applicable,
provided that we can compute it. Its proof is in 
Section~\ref{app:key}.
\begin{lemma}\label{lem:key}
Let $\cA$ be a semisimple $*$-algebra over a field $\F$, $\fdchar(\F)\neq 2$. Let 
$a\in \cA$ be a $*$-symmetric zero-divisor. Then there exists a $*$-symmetric 
element $b\in \cA$, such that $b\Ann_r(a)\not\subseteq a\cA$, where 
$\Ann_r(\cdot)$ denotes the set of right annihilators.
\end{lemma}

Indeed, if $b$ is as in Lemma~\ref{lem:key} in a semisimple $\cA$, 
then viewing 
$a$ and $b$ as 
linear maps 
on $\cA$ (by multiplication from the left), Lemma~\ref{lem:inc_rank} 
gives that we have that for some $\lambda\in S\subseteq 
\F$, 
$|S| > \dim(\cA)$, $\dim((b+\lambda a)\cA)>\dim(a\cA)$.  (When working
with non-semisimple algebras, we 
also make use the simple fact that an element of an algebra is a unit
if and only if it is a unit modulo the radical.)

Thus we wish to work with $\Adj(\vecB')$ and the dimension of
the image of the left
multiplication of its symmetric elements, that is, dimension
of right ideals of 
the 
form $X\Adj(\vecB')$, $X\in L^\epsilon(\vecB')$
-- modulo the radical of $\Adj(\vecB')$. But as $\vecB'$ is not in our 
hand, 
$\Adj(\vecB')$ and $L^\epsilon(\vecB')$ are not either. In fact $\vecB'$ is not 
even 
uniquely determined by $\vecB$. These difficulties can be 
overcome as follows. 
\begin{itemize}
\item For $\Adj(\vecB')$, though $\vecB$ is not $\epsilon$-symmetric, 
we may still define the adjoint algebra 
of $\vecB$ as $\Adj(\vecB)=\{ A\oplus D \in M(n, \F)^{op}\oplus 
M(n, \F) \mid \forall i\in[m], A^tB_i=B_iD\}$. However, while $\Adj(\vecB')$ is 
naturally a $*$-algebra by $(A\oplus D)^*=D\oplus A$, $\Adj(\vecB)$ is not. But 
the 
following relation is easy to verify: $A\oplus D\in \Adj(E \vecB') \iff 
E^{t}AE^{-t}\oplus D\in \Adj(\vecB')$. 
This is because $A^tEB_i'=EB_i'D\iff E^{-1}A^tEB_i'=B_i'D\iff 
(E^tAE^{-t})^tB_i'=B_i'D$.
 So the projection of $\Adj(\vecB)$ to the 
second component coincides with the projection of $\Adj(\vecB')$ to the second 
component. 
\item To get around the lack of $L^\epsilon(\vecB')$ is trickier. We first observe 
that $L^\epsilon(E\vecB F)=F^tL^\epsilon(\vecB)E^{-1}$. Since $\vecB=E\vecB'$, 
$L^\epsilon(\vecB)=L^\epsilon(\vecB')E^{-1}$ so any $Z\in L^\epsilon(\vecB)$ 
equals 
$XE^{-1}$ 
for some $X\in L^\epsilon(\vecB')$. Then consider $XL^\epsilon(\vecB')$: we have 
$XL^\epsilon(\vecB')=XE^{-1}EL^\epsilon(\vecB')=ZL^\epsilon(\vecB'E^t)
=ZL^\epsilon(\epsilon\vecB'^tE^t)=ZL^\epsilon(\epsilon(E\vecB')^t)
=ZL^\epsilon(\epsilon\vecB^t)$. Here we use the assumption that $\vecB'\in 
S^\epsilon(n, \F)^m$.
\end{itemize}
As $L^\epsilon(\vecB')\subseteq \Adj(\vecB')$, 
$L^\epsilon(\vecB')\Adj(\vecB')=\Adj(\vecB')$. Therefore, for any $Z\in 
L^\epsilon(\vecB)$,
$ZL^\epsilon(\epsilon\vecB^t)\Adj(\vecB)=XL^\epsilon(\vecB')\Adj(\vecB')
=X\Adj(\vecB')$ for 
some $X\in L^\epsilon(\vecB')$. Noting that $L^\epsilon(\vecB)$, 
$L^\epsilon(\epsilon\vecB^t)$, and $\Adj(\vecB)$ are what we can compute, this 
allows us to
work with the right ideals of $\Adj(\vecB')$ generated by $X\in L^\epsilon(\vecB')$ without knowing 
the hidden $\vecB'$.

The arguments above lead to the 
following algorithm, assuming that $|\F|>n^2$ and 
$\Rad(\cA)$ can be computed efficiently over $\F$. Fix $S\subseteq \F$ of size 
$>n^2$, 
and perform the following: 
\begin{enumerate}
\item Compute a basis of $L^\epsilon(\vecB)=\langle Z_1, \dots, Z_s\rangle$, 
and 
choose some $Z\in L^\epsilon(\vecB)$. 
\item If $Z$ is full-rank, return $Z$. Otherwise, compute 
$R_Z=ZL^\epsilon(\epsilon\vecB^t)\Adj(\vecB)$.
\item If there exist $i\in[\ell]$ and $\lambda\in S$ such that 
$\dim(R_{\lambda Z+Z_i}+\Rad(\Adj(\vecB)))>\dim(R_Z+\Rad(\Adj(\vecB))$, 
let 
$Z\gets 
\lambda Z+Z_i$ 
and 
go to Step (2). Otherwise 
return ``Not $\epsilon$-symmetrizable''.
\end{enumerate} 
It is clear that the algorithm uses polynomially many 
arithmetic operations, and 
over 
number fields the bit sizes are controlled well. The correctness follows from 
Lemma~\ref{lem:key}, and we only need to exclude a false negative outcome.
Assume to this end that $\vecB$ is $\epsilon$-symmetrizable and
$Z\in L^\epsilon(\vecB)$ is not of full rank. 
With $X=ZE$ and $X_i=Z_iE$,
we have
$R_Z+\Rad(\Adj(\vecB))=
ZL^\epsilon(\epsilon \vecB^t)\Adj(\vecB)+\Rad(\Adj(\vecB))=
XL^\epsilon(\vecB')\Adj(\vecB')+\Rad(\Adj(\vecB'))$,
which is essentially (that is, modulo $\Rad(\Adj(\vecB))$) the image of 
$\ell_{\overline X}$,
where $\overline X$ stands for the residue class of $X$
modulo $\Rad(\Adj(\vecB))$.
As $Z$ is not of full rank, $X$ is not of full rank
either, and hence $X\Adj(\vecB')< \Adj(\vecB')$. Then, as $\Rad(\Adj(\vecB'))$
is the intersection of the maximal right ideals of $\Adj(\vecB')$,
$X\Adj(\vec B')+\Rad(\Adj(\vecB'))<\Adj(\vecB')$.
For $i=1,\ldots, s$, we have that 
$R_{\lambda Z+Z_i}+\Rad(\Adj(\vecB))$ is essentially 
the image of $\lambda \ell_{\overline X}+\ell_{\overline X_i}$. 
Now by Lemma~\ref{lem:key}, if $\ell_{\overline X}$ is not of full
rank, that is $XL^\epsilon(\vecB')+\Rad(\Adj(\vecB'))\neq \Adj(\vecB')$,
or, equivalently, $ZL^\epsilon(\vecB')+\Rad(\Adj(\vecB))\neq \Adj(\vecB)$,
then there exists an element $\lambda$ and a linear combination $b$
of the $\ell_{\overline X_i}$ such that $\lambda \ell_{\overline X}+b$
has larger rank that that of $\ell_{\overline X}$. By linearity, 
$b$ can be chosen from $\ell_{\overline X_i}$ ($i=1,\ldots,s$). 
But then $R_{\lambda Z+Z_i}$ will be, modulo $\Rad(\Adj(\vecB'))$, 
indeed bigger than $R_{Z}$.

\subsection{When $|\F|$ is large enough}\label{app:sub:large_field}

Suppose $|\F|=\Omega(n^4)$.
We 
shall extend the algorithm in 
Section~\ref{subsec:sym_outline} 
to work without relying on the presence of the radical of $\Adj(\vecB)$.
To that end we need some objects to measure ``progress'' modulo
the radical without actually having the radical at hand. These objects
are the right ideals which are, as modules, projective. 
We summarize here definition and the basic
facts known about them, see \cite{Pierce}, Chapter 6, in particular
Section~6.4 for details.

Let $\A$ be an algebra of dimension $d$ with identity.
Projective modules are direct summands of free modules. More specifically, free 
right $\A$ modules are just direct sums of copies of the right $\A$-module 
$\A$ itself, and we say 
that a submodule $M_1$ is a direct summand of the module $M$, if
$M$ is the direct sum of $M_1$ and another submodule $M_2$.
Right ideals that are projective modules are just the direct summands of
$\A$. A right ideal $P$ is projective if it is generated by an idempotent:
$P=e\A$ for some idempotent $e\in \cA$.
(This is equivalent 
to saying that
$e$ is a left identity element of $P$.) Every projective right $\A$-module
$P$ can be decomposed into the direct sum of indecomposable projective
modules. A projective module $P$ is indecomposable, if and only if its {\em head}
$P/P\Rad(\A)$ 
is a simple $\A/\Rad(\A)$-module. 
Two indecomposable 
projective modules
are isomorphic if and only if their heads are isomorphic. Projective
indecomposable modules are also called principal indecomposable, and they
appear as projective right ideals of $\A$. Every projective right $\A$-module
can be decomposed into a direct sum of principal indecomposable
modules. The decomposition is, up to the isomorphism types of the
principal indecomposable with multiplicities, is unique. This theory,
when specialized 
to right ideals of $\A$, gives
that right ideals that are projective as modules (for brevity, 
we will call them projective right ideals) can be decomposed into direct 
sums of indecomposable  right ideals, which are also projective.
The multiplicities of the various
principal indecomposable modules in $P$ are the same as the 
number of various simple components of the factor
$P/\Rad(\A)P$ as a right module over the semisimple
algebra $\A/\Rad(\A)$. Note that for a projective right ideal $P$,
we have $\Rad(\A)P=P\cap \Rad(\A)$, which can be shown easily using
a generating idempotent. It follows that $P/\Rad(\A)P\cong (P+\Rad(\A))/\Rad(\A)$.
Here we have that $(P+\Rad(\A))/\Rad(\A)$ is a right ideal
of $\A/\Rad(\A)$. If $\A/\Rad(\A)$ is the direct sum of simple
algebras $\A_1,\ldots,\A_\ell$, then $(P+\Rad(\A))/\Rad(\A)$ is decomposed into 
the 
direct sum $P_1,\ldots,P_\ell$ where $P_i=((P+\Rad(\A))/\Rad(\A))\cap 
\A_i$.
Now
each $P_i$ can be decomposed into a sum of minimal right ideals of $\A_i$, where
the number of such components are the multiplicities of the various
principal indecomposables in the decomposition of $P$.

Here we show that, given a right ideal $J$ of $\A$, a projective
right ideal $P$ can be computed with the property that 
$(P+\Rad(\A))/\Rad(\A)=(J+\Rad(\A))/\Rad(\A)$. By the discussion above,
the module structure of $P$ depends only on the factor $(J+\Rad(\A))/\Rad(\A)$,
and hence can be used to (partially) compare right ideals modulo the
radical. The key property is the equivalent characterization of
projective right ideals as those generated by idempotents of $\A$,
see \cite{Pierce}, Section~6.4.

\begin{proposition}\label{prop:jzero}
Let $\A$ be a finite dimensional algebra
with identity and let $J$ be a non-nilpotent right ideal of $\A$.
Then in deterministic polynomial time one can compute a
right ideal $J_0$ contained in $J$ generated by
an idempotent $e$ such that $e+\Rad(\A)$ is a left identity element of
$(J+\Rad(\A))/\Rad(\A)$. 
\end{proposition}
\begin{proof}
To compute
$J_0$, it is sufficient to find an idempotent $e$ of $J$ with the
property as in the statement. As $J$ is not nilpotent, one
can find a non-nilpotent element and even an idempotent $e$ 
in $J$, as shown in Fact~\ref{fact:non-nil} and the remark following its proof. 
Compute the right ideal 
$J''=\{x-ex:x\in J\}$. Obviously
$e\A\cap J''=0$. If $J''$ is nilpotent, then $e$ is as requested. 
Otherwise find an idempotent $f$ in $J''$. We have $ef=0$ and
$(e+fe)^2=ee+fefe+efe+fee=e+fe$. So if $fe\neq 0$, then
we can replace $e$ with $e+fe$ which generates a right ideal
larger than $e\A$. If $fe=0$, then $(e+f)^2=e+f$ whence we 
can proceed with $e+f$ in place of $e$. 

Over a number 
field, some care is needed 
to ensure that size of the data representing
the idempotent $e$ do not explode. In order to do this,
we fix a basis for $J$ and express $e$ in terms of that basis.
Let $n=\dim J$. 
We consider the matrix representation of $J$
on itself by action from the left. Then for
every element $x\in J$, let
$N_x=\{v\in J:x^nv=0\}$ be the generalized $0$-eigenspace
of $x$. We have $J=x^nJ\oplus N_x$, and there is an idempotent
$e_x$ in the subalgebra generated by $x^n$ with $e_xJ=x^nJ$.
Assume that we have an idempotent $e\in J$ at hand with $\dim(eJ)=r$.
Considering $e$ as an $n$ by $n$ matrix, we have that $e=e^n$ has
rank $r$: there is an $r$ by $r$ submatrix whose determinant of $e^n$ 
is nonzero. We consider this determinant for the $n$th power 
of the matrix of 
 a generic element $x$ form $J$. This determinant has degree at most
$nr<n^2$ in the coordinates of $x$. We have at hand $e$, that is
a specific assignment for the coordinates on which this polynomial
takes a nonzero value. Given a subset $\Lambda$ of size $n^2$ of $\F$,
we can replace the first coordinate of $e$ by
an element of $\Lambda$ such that the for the new element $x$,
its power $x^n$ has rank at least $r$. Then we can proceed with the second
coordinate, and so on, see~\cite[Lemma 2.2]{dGIR96} for a formal statement.
When finished, we have an element $x$ of small coordinates such that
$x^n$ has still rank at least $r$. Now the identity element $e_x$ of the
subalgebra of $J$ generated by $x^n$ will have the same rank
and still moderate coordinates. (This algebra
is spanned by $x^n,x^{2n},\ldots,x^{n^2}$.)
We replace $e$ with $e_x$ and continue increasing its rank 
if $(1-e)J$ is not nilpotent.
\end{proof}

We call the right ideal $J_0$ as in Proposition~\ref{prop:jzero} 
the projective module associated to $J$, and denote it by $P(J)$. 
For a nilpotent right ideal $J$ we set $P(J)=0$. As the decomposition
of $P(J)$ into principal indecomposables reflects faithfully the decomposition
of $(J+\Rad(\A))/\Rad(\A)$ into simple modules, the map $J\mapsto
P(J)$ is \emph{monotone} in $J$ modulo the radical, in the following sense: if 
$(J+\Rad(\A))/\Rad(\A)$
is isomorphic to a proper submodule of $(J'+\Rad(\A))/\Rad(\A)$, then
$P(J)$ is isomorphic to a proper submodule of $P(J')$.

\begin{fact}\label{fact:conj} 
Let $\A$ be a finite dimensional semisimple algebra with identity,
let $\A_1,\ldots,\A_\ell$
be the simple components of $\A$, and let $\pi_j:\A\rightarrow \A_j$, 
$j\in[\ell]$, 
be the corresponding projections. Suppose that 
$e$ and $f$ are idempotents in $\A$ such that the rank of $\pi_j(e)$
is the same as that of $\pi_j(f)$ for $j=1,\ldots,\ell$. Then
$e\A$ and $f\A$ are isomorphic as right $\A$-modules.
\end{fact}
\begin{proof}
Indeed, for each individual $j$, $\pi_j(e\A)$ and $\pi_j(f\A$), respectively,
 are 
direct sums of minimal right ideals of $\A_j$, which are, as modules,
isomorphic copies of the same simple right $\A$-module $S_j$. Here we used
the fact that, for the simple algebra $\A_j$, up to isomorphism there is one
simple right module, which is present in $\A_j$ as a minimal right ideal. 
The multiplicity of $S_j$ in the decomposition of $e\A$ (resp. $f\A$)
is then just the quotient of the rank of $\pi_j(e)$ (resp. $\pi_j(f)$) by
the dimension of $S_j$.
\end{proof}

Now we are ready to upgrade the algorithm in Section~\ref{subsec:sym_outline} to 
work 
without knowing the radical of 
$\Adj(\vecB)$. 
\begin{proposition}\label{prop:large_field}
Let $\A$ be an $m$-dimensional algebra with identity,
let $a$ be a zero-divisor in $\A$ and let $b\in \A$
such that $a+\Rad(\A)$ and $b+\Rad(\A)$ behave like
$a$ and $b$ in Lemma~\ref{lem:key}. If $S$ is a subset 
of the base
field of size $\Omega(m^2$), then for at least one $\lambda\in S$ we have
$\dim P((\lambda a+b)\A)>\dim P(a\A)$.
\end{proposition}
\begin{proof}
We have that, modulo $\Rad(\A)$,
$\lambda a+b$ generates a right ideal that has
dimension higher than the one generated by $a$
for at least one field element $\lambda$ from $S$ if 
$|S|=\Omega(m)$. If $S$ is even larger, say $|S|=\Omega(m^2)$, then $S$ 
will contain
such a field element $\lambda$ with the additional property that, the projection 
of $(\lambda 
a+b)\A$ to any of
the simple components of $\A/\Rad(\A)$ has dimension
at least as high as that for the projection of
$a\A$. 
The existence of such a $\lambda$ is 
ensured by applying Lemma~\ref{lem:inc_rank} iteratively to the projections to the 
simple components.
Then, by Fact~\ref{fact:conj}, the right
$\A$-module $a\A+\Rad(\A)$
can be embedded  into $(\lambda a+b)\A+\Rad(\A)$ 
as a proper submodule.
By monotonicity as explained in the paragraph before Fact~\ref{fact:conj}, 
$P(a\A)$ is isomorphic to a proper submodule
of $P((\lambda a+b)\A)$.
\end{proof}

\subsection{When $|\F|$ is small}\label{app:sub:small_field}

The algorithm in Section~\ref{subsec:sym_outline}, upgraded in 
Section~\ref{app:sub:large_field}, runs in polynomial time even over a number 
field, but 
has the disadvantage of relying on the field to be large enough. In this 
subsection, we present an algorithm that works even for small fields. However, the 
disadvantage of this algorithm is that, over a number field it seems difficult to 
bound the 
bit sizes of intermediate data. Still, combining these two algorithms together we 
are able to cover all fields, so this proves Theorem~\ref{thm:sym}. 

As explained in Section~\ref{subsec:sym_outline}, w.l.o.g. we can assume $\vecB$ 
to be 
non-degenerate. The following Lemma~\ref{lem:star1} is the key to this algorithm. 
Its proof is put in Section~\ref{app:key}. 

\begin{lemma}\label{lem:star1}
Let $\F$ be a field of characteristic not $2$. Let $\A$ be a finite dimensional
$*$-algebra over $\F$ with
an identity element. Let $a$ be a $*$-symmetric
element of $\A$ such that the right ideal $a\A$ has a left 
identity element. Then the right annihilator 
$\Ann_r(a)=\{b\in \A:ab=0\}$ of $a$
is generated, as a right ideal, by a
$*$-symmetric
 element of $\A$.
\end{lemma}

We remark that the condition that $a\A$ has a left identity element is
just equivalent to that $a\A$ is projective as a right $\A$-module;
see Section~\ref{app:sub:large_field} for this concept.

We shall only sketch the idea behind the algorithm in the following; a rigorous 
algorithm can be extracted without much difficulty. 

Suppose $\vecB=E\vecB'$ where $E\in\GL(n, 
F)$ and $\vecB'\leq S^\epsilon(n, \F)$. We claim that
$L^\epsilon(\vecB')$ cannot be spanned by nilpotent
elements. Indeed, assume the contrary. Let
$\A=\Adj(\vecB')$, which is a $*$-algebra as $\vecB'$ is $\epsilon$-symmetric. 
Then 
$I\otimes *$ is
an involution of $\overline \A={\overline \F}\otimes_\F \A$,
where $\overline \F$ is an algebraic closure
of $\F$. We identify $\A$ with the subalgebra $1\otimes \A$
and use $*$ for $I\otimes *$. The $*$-symmetric elements of $\overline \A$
are $\overline \F$-linear combinations of $*$-symmetric
elements of $\A$. Using this, we may assume that $\F$ is
algebraically closed.  
Then the $*$-simple components of the factor
of $\A/\Rad(\A)$ contain $*$-symmetric idempotents whose images
are rank one or two matrices 
under some irreducible representation of $\A$.
It follows that any basis for $L^\epsilon(\vecB')$ contains
an element whose image under a matrix representation
of $\A$ has nonzero trace. Such an element cannot be 
nilpotent. 

Thus any basis of $L^\epsilon(\vecB)=L^\epsilon(\vecB')E^{-1}$ contains
an element of the form $Z=XE^{-1}$ where $X$ is a non-nilpotent
element of $L^\epsilon(\vecB')$. Now consider the subspace 
$XL^\epsilon(\vecB')X$. This set equals the set of the 
$*$-symmetric elements of the subalgebra $X\A X$. This subalgebra is not 
nilpotent, as it contains the 
non-nilpotent
element $X$. 
Therefore, just like above,
an arbitrary basis for $XL^\epsilon(\vecB')X$ contains
a non-nilpotent element. 
It follows that
an arbitrary basis 
for $L^\epsilon(\vecB')$
(which may differ from the basis which $X$ is chosen from) 
contains an element $Y$ such that $XYX$ is not nilpotent.
In particular, a basis for $L^\epsilon(\epsilon \vecB^t)=EL^\epsilon(\vecB')$
contains an element $Z'$ of the form $Z'=EY$ where $XYX$ is
not nilpotent. Now consider the sequences $X_k=X(YX)^k$
and $Y_k=Y(XY)^k$, $k\geq 0$. We have $X_0=X$, $Y_0=Y$,
$X_{k+1}=XY_kX$, and $Y_{k+1}=YX_kY$. Furthermore $X_{k+1}E^{-1}=
(XE^{-1})(EY_k)(XE^{-1})$ and $EY_{k+1}=(EY)(X_kE^{-1})(EY)$,
which gives an efficient method for computing $X_kE^{-1}$
and $EY_k$. The kernels of $X_k$ form a nondecreasing chain
of linear spaces. Therefore if $k$ is large enough, then
$\ker X_\ell=\ker X_k$ for $\ell>k$. The sequences consisting
of the kernels of $Y_k$, as well as those consisting
of the images of $X_k$ and the images of $Y_k$, stabilize
as well. From
$Y_{2k+1}=YX_kY_k$, 
we infer that
for sufficiently large
$k$ the kernel of $X_kY_k$ is the same as
that of $Y_k$, 
and the image of $X_kY_k$ is the same as
that of $X_k$. 
Analogous equalities hold for the kernel
and for the image of $Y_kX_k$. These properties of the pair
$X_k,Y_k$ imply that the image of $Y_k$ is 
a direct complement of the kernel of $X_k$, and the image
of $X_k$ is a direct complement of the kernel of $Y_k$.

As $X_kY_k=X_kE^{-1}EY_k$, we can efficiently compute the
product $X_kY_k\in \Adj(\vecB')$, which cannot be zero. Note that if $X_kY_k$ 
is 
invertible, then $X$ is also invertible, and the $XE^{-1}$ in our hand sends 
$\vecB$ 
to $\vecB'$, which solves the problem. So in the following we assume $X_kY_k$ has 
a 
non-trivial kernel. 

Similarly to the stabilization
argument above, we may assume that $k$ is large enough so that the kernel of 
$X_kY_k$ in the left regular representation $\Adj(\vecB')$
is a direct complement of the image. This mean that the right
annihilator of $X_kY_k$ in $\Adj(\vecB')$ (which is the same as
that of $Y_k$) 
 and the right 
ideal
generated by $X_kY_k$ (which is also generated by 
$X_k$)
are complementary to each other and the same holds for
the product $Y_kX_k$.

We claim that there exists 
$\vecB''\leq S^\epsilon(n, \F)$ such that
$\vecB=E'\vecB''$ for some invertible $E'$ and $X_kY_k\in L^\epsilon(\vecB'')$.
To see this, consider an element $Z\in L^\epsilon(\vecB')$ 
which is a generator of the right annihilator
of $X_k$ as a right ideal in $\Adj(\vecB')$. Such $Z$ 
exists by Lemma~\ref{lem:star1}. Put $W=Y_k+Z$. 
Then $W\in L^\epsilon(\vecB')$, and $W$ is invertible since 
$Y_k$ and $Z$ are generators
of right ideals of $\Adj(\vecB')$ complementary to each other. We also
have $X_kW=X_k(Y_k+Z)=X_kY_k$. Let $\vecB''=W^{-1}\vecB'$. Then,
$W^{-1}$ is an invertible element of $L^\epsilon(\vecB')$, so we have
$\vecB''\leq S^\epsilon(n, \F)$. Furthermore, $L^\epsilon(\vecB'')=
L^\epsilon(W^{-1}\vecB')=L^\epsilon(\vecB')W$. In particular, $X_kY_k=X_kW\in
L^\epsilon(\vecB')W=L^\epsilon(\vecB'')$.

Let $J$ (resp.~$K$) be the image (resp.~the kernel)
of $X_kY_k$. From $X_kY_k\in L^\epsilon(\vecB'')$
we infer $J=K^{\perp_{\vecB''}}$. Let
$J'=K^{\perp_\vecB}$ and $K'=J^{\perp_\vecB}$.
These subspaces can be computed efficiently.
Let $U_0$ be an invertible linear map
that maps $J$ to $J'$ and $K$ to $K'$.
Then by replacing $\vecB$ with $U_0^t\vecB$ we can arrange that
$J=K^{\perp_\vecB}$ as well. Then the problem
can be reduced to the subspaces $J$ and $K$.

\subsection{Two lemmas about $*$-algebras}\label{app:key}

For the next two lemmas, we depend crucially on the structure of $*$-algebras as 
described in the first paragraph of ``Structure 
of $*$-algebras'' in Section~\ref{sec:prel}. We begin with a claim
which is used in the proofs of both.

\begin{claim}
\label{claim:centralizer}
Let $\A$ be a semisimple $*$-algebra, and let $a\in \A$ be 
a $*$-symmetric zero divisor, such that there is no proper $*$-symmetric
idempotent $e$ with $ae=ea$. Then $a$ is nilpotent and $\A$ is $*$-simple.
Furthermore, either 
\begin{enumerate}
\item[(i)]
$\A\cong M(n, \D)$ for a division algebra $\D$ and
$a$, as an $n$ by $n$ matrix over $\D$ has just one
Jordan block;
\item[(ii)] There exists a proper idempotent $f\in\A$ with
$af=fa$ and $f^*=1-f$, $f\A f\cong M(n, \D)$ for some
division algebra $\D$ such that $faf$,  
as an $n$ by $n$ matrix over $\D$ has just one
Jordan block.
\end{enumerate}
\end{claim}
We remark that case (ii) captures two sub-cases:
either $f$ is in the center of $\A$, and $\A$ is
of exchange type; or $\A\cong M_{2n}(\D)$ and
$a$, as an $2n$ by $2n$ matrix, has two Jordan blocks 
of size $n$.

\begin{proof}
Let $\Calg$ be the centralizer of $a$ in $\A$, that is, $\Calg=\{x\in \A:xa=ax\}$.
Then it is straightforward to see that $\Calg$ is a $*$-subalgebra of
$\A$ containing $a$. Then, as subalgebras generated
by non-nilpotent zero divisors do contain nontrivial idempotents,
$a$ is nilpotent. Furthermore, as the center of $\A$ is contained
in $\Calg$, there are no $*$-invariant central idempotents in $\A$. In other 
words, $\A$ is $*$-simple: it is either simple or
of exchange type consisting of two simple components. 

Furthermore, every $*$-symmetric element of $\Calg$ is either
nilpotent or invertible, that is, $\Calg$ is Osborn-local, 
whence $\Calg/\Rad(\Calg)$ is an Osborn-division algebra. By
Osborn's theorem \cite[Theorem 2]{Osb70}, $\Calg/\Rad(\Calg)$ cannot contain 
three or more
pairwise orthogonal idempotents, and if it contains 
any proper idempotent then it also contains a proper idempotent
${\overline f}$ with ${\overline f}^*=1-{\overline f}$.
We claim that in the latter case there exists a proper
idempotent $f$ in $\Calg$ such that $f^*=1-f$. Such
an $f$ can be constructed using the following iteration.
Let $J$ be an ideal of $\Calg$ contained in $\Rad(\Calg)$,
and $f$ be an element of $\Calg$ such that $f^2-f\in J$
and $f+f^*-1\in J$. Initially, $J=\Rad(\Calg)$ and $f$
is an arbitrary element of the coset ${\overline f}$. 
The iterative step starts with arranging that $f^2-f\in J^2$:
this can be done by any standard lifting technique,
e.g., by replacing $f$ with $3f^2-2f^3\in f+J$
(see \cite{DrozdKirichenko}, proof of Lemma 3.2.1).
Then, as this new $f$ is from the same residue class modulo $J$
as the old one, the property $f+f^*-1\in J$ is preserved.
Next we put
$r=\frac{1}{2}(f+f^*-1)\in J$ 
and 
$f'=f-r$. Then we have $f'+f'^*-1=f+f^*-2r-1=0\in J^2$.
Furthermore, ${f^*}^2-f^*\in J^2$ can be rewritten as
$((1-f)+2r)^2-(1-f)-2r\in J^2$, which 
implies $2r-2fr-2rf+(f^2-f)+4r^2\in J^2$,
whence $-r+rf+fr\in J^2$.
It follows that ${f'}^2-f'=f^2-rf-fr+r^2-(f-r)= 
(f^2-f)-(rf+fr-r)+r^2\in J^2$. 
Therefore, we can proceed
with $f'$ in place of $f$ and $J^2$ in place of $J$.
Note that if $\A$ is of exchange type, then $f$ can be
even chosen as the identity element of one of the simple
components of $\A$.  In any case, the subalgebra $f\A f$ 
must be simple, so isomorphic to $M(n, \D)$ for some $n$ and 
$\D$. Assume that $faf$ has more than one Jordan blocks.
Let $e\in f\A f$ be the block diagonal
matrix which is the identity in one of the Jordan blocks
of $a$ and zero elsewhere. Then $e$ is an idempotent commuting with
$a$ with $ef=e\neq f$. Then $e^*\in f^*\A f^*$ and
$e+e^*$ is a proper idempotent commuting with $a$. So
$faf$, as an $n$ by $n$ matrix over $\D$, must consist of 
single (nilpotent) Jordan block of size $n$.

If $\Calg$ does not contain proper idempotents then
$\A$ is itself simple, isomorphic to $M(n, \D)$ and
again, $a$ has just one (nilpotent) Jordan block of
size $n$.
\end{proof}

%
\vskip .5em
\noindent{\bf Lemma~\ref{lem:star1}, restated.} Let $\F$ be a field of 
characteristic not $2$. Let $\A$ be a finite dimensional
$*$-algebra over $\F$ with
an identity element. Let $a$ be a $*$-symmetric
element of $\A$ such that the right ideal $a\A$ has a left 
identity element. Then the right annihilator 
$\Ann_r(a)=\{b\in \A:ab=0\}$ of $a$
is generated, as a right ideal, by a
$*$-symmetric
 element of $\A$.
 
\begin{proof}
Note that $e\in\A$ is a left identity element of the right ideal
$a\A$ if and only if $ea=a$ and there exists $d\in \A$ such that 
$e=ad$. Let $e$ be such an element. 
Then $e^*=d^*a$ is a right
identity element of the left ideal $\A a$. We claim
$\Ann_r(a)$ is the right ideal of $\A$ generated by 
the idempotent $1-e^*$. Indeed, assume that $ab=0$.
Then $e^*b=d^*ab=0$. Conversely, if $e^*b=0$ then
$ab=ae^*b=0$. Thus $b\in \Ann_r(a)$ if and only if
$e^*b=0$. The latter equality  is equivalent to that 
$b=(1-e^*)b'$ for some $b'\in \A$.

Next we show that we may assume that $\A$ is semisimple.
To see this, let $e$ be an idempotent as above. Then
$\Ann_r(a)=(1-e^*)\A$. Let $\phi$ be the projection
$\A\rightarrow {\overline \A}:=\A/\Rad(\A)$. We denote
the involution of $\overline \A$ induced 
also by $*$. 
Obviously, $\phi(e)$ is an idempotent in 
$\phi(a){\overline \A}$ with $\phi(e)\phi(a)=\phi(a)$.
It follows that the right annihilator of $\phi(a)$
is generated by $1-\phi(e^*)$, whence it coincides
with $\phi(\Ann_r(a))$. Similarly, the left annihilator
of $\phi(a)$ is the left ideal of $\Phi(\A)$
generated by $(1-\phi(e))$ and it coincides with
$\phi(\Ann_l(a))$. It follows that 
\begin{equation}\label{eq:two_side}
\Ann_r(\phi(a))\cap \Ann_l(\phi(a))=
\phi(\Ann_r(a)\cap \Ann_l(a)).
\end{equation}
(The annihilators on the
left hand side are understood as inside $\overline \A$.)
Assume that the assertion holds in ${\overline \A}$. 
Then there is an element $\overline b$ of $\overline \A$
such that ${\overline b}^*={\overline b}$, and 
${\overline b}$ generates the right annihilator of $\phi(\A)$
in ${\overline \A}$. Notice that ${\overline b}$ annihilates
$\phi(a)$ from the left as well. Therefore, by Equation~\ref{eq:two_side}, 
${\overline b}$ has 
a preimage $b$ in $\Ann_r(a)\cap \Ann_l(a)$. We have
$b-b^*\in \Ann_r(a)\cap \Ann_l(a)\cap\Rad(\A)$. Therefore,
by replacing $b$ with $b-\frac{1}{2}(b-b^*)$ we can arrange
that $b^*=b$. We have $b\A+\Rad(\A)=\Ann_r(a)+\Rad(\A)$.
The isomorphism theorem, applied to the linear spaces 
$b\A+(\Ann_r(\A)\cap \Rad(\A))$, $\Ann_r(a)$, and
$\Rad(\A)$, gives
$b\A+(\Ann_r(\A)\cap \Rad(\A))/(\Ann_r(a)\cap \Rad(\A)\cong
\Ann_r(a))/(\Ann_r(a)\cap \Rad(\A))$. It follows that 
$b\A+(\Ann_r(a)\cap \Rad(\A))=\Ann_r(a)$. {}From
$\Ann_r(a)=(1-e^*)\A$, we infer that the radical
of $\Ann_r(a)$ as a right $\A$-module is
$\Ann_r(a)\Rad(\A)=(1-e^*)\Rad(\A)=\Ann_r(a)\cap\Rad(\A)$. Thus
$b$ generates the right $\A$-module $\Ann_r(a)$ modulo its radical,
whence $b\A=\Ann_r(a)$. Therefore we may indeed 
assume that $\A$ is semisimple. Furthermore, by going over 
the $*$-simple components, we can assume that
$\A$ is even $*$-simple, that is, 
$\A$ is either a simple algebra or a direct sum $\B\oplus \B^{\tiny\rm op}$ where
$\B$ is a simple algebra and $(\beta,\beta')^*=(\beta',\beta)$. 

Assume that $\A=\B\oplus \B^{\tiny\rm op}$. Then $a$ is of the form 
$(\alpha,\alpha)$ and $\Ann_r(a)$ consists of pairs $(\beta,\beta')$
where $\beta\in \Ann_r(\alpha)$ and $\beta'\in\Ann_l(\alpha)$. Let
$\delta$ be an invertible element of $\B$ such that $\alpha\delta$
is an idempotent. Then $1-\alpha\delta$ is 
a generator for the right annihilator and for the left annihilator of 
$\alpha\delta$ inside $\B$ at the same time. Note that the latter is the same as 
the 
left annihilator
of $\alpha$. Put $\gamma=\delta(1-\alpha\delta)$. Then $\alpha\gamma=0$
and $\gamma\alpha=0$. Also, the dimensions of the one-sided ideals generated 
by $\gamma$ are the same as those generated by 
$1-\alpha\delta$. 
Therefore $\gamma$ generates as one
sided ideals both the left and the right annihilators of $\alpha$ inside
$\B$. It follows that $(\gamma,\gamma)$ is a generator for $\Ann_r(a)$
as a right ideal of $\A$.

The rest of the proof is for the case where $\A$ is a simple algebra:
$\A\cong M(n, \D)$, where $\D$ is a division algebra. Note that in 
the (semi-)simple case every one-sided ideal is generated by an
idempotent. 

We first consider the following case: suppose we have $f\in\A$ which is a proper 
idempotent in  
$\A$ such that $f^*=f$ and $fa=af$. If $b\in \Ann_r(a)$ 
then $afb=fab=0$ and $a(1-f)b=(1-f)ab=0$, whence 
$\Ann_r(a)$ is decomposed into the direct sum 
of $f\Ann_r(a)=f\A\cap \Ann_r(a)$ and 
$(1-f)\Ann_r(a)=(1-f)\A\cap \Ann_r(a)$.
From the fact that,
 in a simple algebra, an arbitrary right ideal $J$ is generated 
by the subspace $Jg$ for any nonzero idempotent element $g$, we infer
that $f\Ann_r(a)=f\Ann_r(a)f\A$. We claim that $f\Ann_r(a)f$
is the right annihilator of $faf$ in the subalgebra $f\A f$.
Indeed, if $ab=0$ then $faffbf=fabf=0$, demonstrating 
$f\Ann_r(a)f\subseteq f\A f\cap \Ann_r(faf)$. To see the reverse
inclusion let $fbf\in \Ann_r(faf)$. Then $0=fafbf=afbf$, whence
$fbf\in \Ann_r(a)$ and $fbf=f^2bf^2\in f\Ann_r(a)f$. Assume by
induction that the statement of the lemma holds in the simple $*$-invariant
subalgebra $f\A f$. Then there exists an element $fb_1f$ with
$fb_1^*f=fb_1f$ generating $f\Ann_r(a)f$ as a right ideal
of $f\A f$. Then by the discussion above, the right ideal
of $\A$ generated by $fb_1f$ is $f\Ann_r(a)$. Similarly,
we can use induction to show the existence of $b_2$ with
$(1-f)b_2(1-f)=(1-f)b_2^*(1-f)$ such that $(1-f)b_2(1-f)$
generates the right ideal $(1-f)\Ann_r(a)$. Then
the element $b=fb_1f+(1-f)b_2(1-f)$  is $*$-symmetric
generator for the right ideal $\Ann_r(a)$.

We then consider the case
when there are no proper 
$*$-symmetric idempotents in the simple algebra $\A\cong M(n, \D)$
commuting with $a$.
Then, by Claim~\ref{claim:centralizer} $a$ is nilpotent,
and, as an $n$ by $n$ matrix over $\D$, 
either has one Jordan block of size $n$, or has two
Jordan blocks of $\frac{n}{2}$. Then
the annihilator of $a$ is generated by $a^{n-1}$ or by
$a^{\frac{n}{2}-1}$, respectively.
\end{proof}

\noindent{\bf Lemma~\ref{lem:key}, restated.} Let $\cA$ be a semisimple 
$*$-algebra over a field $\F$, $\fdchar(\F)\neq 2$. Let 
$a\in \cA$ be a $*$-symmetric zero-divisor. Then there exists a $*$-symmetric 
element $b\in \cA$, such that $b\Ann_r(a)\not\subseteq a\cA$, where 
$\Ann_r(\cdot)$ denotes the set of right annihilators.

\begin{proof}
Assume that there exists a proper idempotent $f$ in $\A$
such that $f^*=f$ and $fa=af$. Then either $fa=af=faf$ is a zero-divisor
in $f\A f$, or $(1-f)a(1-f)$ is a zero-divisor in $(1-f)\A(1-f)$.
(For, if $c_1\in f\A f$ such that $f=fafc_1$, and  $c_2\in (1-f)\A(1-f)$
such that $(1-f)=(1-f)a(1-f)c_2$, then from $fc_1=c_1$ and $fc_2=0$ we
infer that $a(c_1+c_2)=
((faf)+(1-f)a(1-f))(c_1+c_2)=fafc_1+(1-f)a(1-f)c_2=f+(1-f)=1$.)
Assume that $faf$ is a zero-divisor in $f\A f$. Then by induction,
there exist $b_1,c_1\in f\A f$
such that $b_1^*=b_1$, $faf c_1=0$ and $b_1c_1\not \in faf f\A f$.
We have $c_1\in \Ann_r(a)$ because 
$0=fafc_1=af^2c_1=afc_1=ac_1$.
We claim
that $b_1c_1\not \in a\A$. Indeed, assume that $b_1c_1=ad$ for some
$d\in \A$. Then $fb_1=b_1$ and $c_1f=c_1$ imply $b_1c_1=fadf$ and
$fadf=faffdf\in
faf f\A f$, a contradiction with the assumption. 

Based on the above, it is sufficient to prove the assertion
when $\A$ has no proper $*$-symmetric idempotents commuting
with $a$. Then, 
by Claim~\ref{claim:centralizer},
$\A$ is $*$-simple and $a$ is nilpotent. Furthermore,
either $\A\cong M(n, \D)$ for some division algebra
$\D$ and $a$ has a single Jordan block of size $n$,
or there exists an
idempotent $f$ of $\A$ such that $f^*=1-f$, $fa=af$,
$f\A f\cong M(n, \D)$ for some division algebra $\D$, 
and $fa=faf$ has just one Jordan block. The latter case
covers two cases from the point of view of the structure
of $\A$: it can be either simple or a sum of two simple
components, corresponding to the case whether or not $f$ is
central in $\A$.

In the latter case we consider an isomorphism 
$\phi:f\A f\rightarrow M(n, \D)$ such that $\phi(fa)$ 
is in Jordan normal form. Then let $b_1\in f \A f$ such
that $\phi(b_1)$ is everywhere zero except in the lower
left corner. Then $b=b_1+b_1^*$ will do.

When $\A\cong M(n, \D)$ and 
$a$ has a single Jordan block of size $n$,
to prove the lemma, it is enough to show the following: there exists an 
appropriate basis for 
$\D^n$ such that the following holds. First, $a$, as an $n$ by $n$ matrix, is of 
Jordan normal form. Second, there 
is a matrix that is everywhere zero,
except at the lower left corner corresponding to 
a $*$-symmetric element of $\A$.
To this end, let $f$ be a right identity element
of the left ideal $\A a^{n-1}$. Then $f$ is a primitive
idempotent in $\A$, and $f^*$ is a left identity
element of the right ideal $a^{n-1} \A$. As $\A a^n=a^n \A=0$,
we have $fa=af^*=0$ and $ff^*=0$. We claim that we can arrange that
$f^*f=0$ as well. Indeed, setting $f'=f-f^*f$,
we have ${f'}^*f'=f^*f-f^*f f^*-f^*f f+f^*f f^*f=f^* f-0-f^*
f-0=0$, $f'=(1-f^*)f\in \A f$, and $ff'=f$. The latter two properties
show that $f'$ is an identity element of $\A f$. We replace $f$ with $f'$.
Then $f$ and $f^*$ are orthogonal primitive idempotents.

The subspace $f\A f^*$ is $*$-invariant. Assume that
$f\A f^*$ does not contain nonzero $*$-symmetric elements.
Then for every $c\in f^*\A f$ we have $c^*=-c$ (other wise $c^*+c$ would be 
nonzero and 
$*$-symmetric for some nonzero $c$.) Put $g=f+f^*$. Then the subalgebra $\A'=g\A 
g$ 
is isomorphic to $M(2, \D)$. Let $\phi:g\A g \mapsto M(2, \D)$
be an isomorphism that maps $f$ and $f^*$ to the block diagonal
idempotent matrices $\mbox{diag}(0,1)$
and to $\mbox{diag}(1,0)$, respectively. Then $\phi(a^{n-1})$
is a matrix which is nonzero exactly at the upper right corner, and
$\phi(f \A f^*)$ consist of matrices whose entries are all zero
except possibly that at the lower left corner. It
follows that there exists element $c\in f^*\A f$ such that $\phi(c)$
is nonzero only at the lower left corner, where the entry is
the inverse (in $\D$) of the upper right entry of $\phi(a^{n-1})$.
For this $c$ we have $a^{n-1}c=f^*$ and $ca^{n-1}=f$. It follows
that the subspace spanned by $f^*$, $f$, $a^{n-1}$ and $c$
form a subalgebra. (It is actually isomorphic to the algebra
of the $2$ by $2$ matrices over the base field.) The assumption
$c^*=-c$ implies that this subalgebra is $*$-invariant.
However, it is straightforward to verify that the restriction
of $*$ does not give an involution on this subalgebra: $(ca^{n-1})^*=f^*\neq 
-f^*=-a^{n-1}c$.

Thus there exists a nonzero $*$-symmetric element $c\in f\A f^*$.
We have $a^{n-1}\in \Ann_r(a)$ and $ca^{n-1}\in f\A f\setminus \{0\}$.
Then $ca^{n-1}\not\in a\A$.
\end{proof}

\appendix

\section{Comparison with the result of Berthomieu et al. 
\cite{BFP15}}\label{app:compare}

Recall that in \cite{BFP15}, the algorithm works under the two conditions: (1) 
there exists a non-degenerate form in the linear span of the given form, and (2) 
the underlying field is large enough. The algorithm needs to find a solution 
possibly from an extension field. To compare our algorithm with theirs, we first 
present an algorithm that works under conditions (1) and (2) but does not require 
going over an extension field. It follows the general principle of our algorithm 
in 
Section~\ref{sec:iso} and suggests the role of a hidden $*$-algebra. This allows 
us to explain what the algorithm of \cite{BFP15} is like, and why our algorithm 
avoids using extension fields. 

Suppose we are given two tuples of symmetric matrices $\vecB=(B_1, \dots, B_m)$ 
and $\vecC=(C_1, \dots, C_m)$, $B_i, C_j\in M(n, \F_q)$ where $q$ is an odd prime 
power. We aim to find $X\in\GL(n, \F_q)$ such that $\forall i\in[m]$, 
$X^tB_iX=C_i$. 
The regularity condition (1) as well as the field size condition (2) in 
\cite{BFP15} imply that we can compute and 
therefore assume w.l.o.g. $B_1$ is non-singular. Therefore $B_1$ and $C_1$ must be 
isometric and we can transform $C_1$ to $B_1$ using techniques from 
Section~\ref{subsubsec:VI},
so in the 
following we assume $B_1=C_1$. 

Now define an involution $*$ on $M(n, \F_q)$ by $A^*=B_1^{-1}A^tB_1$. Note that 
$A^t=B_1A^*B_1^{-1}$. Then the equation $X^tB_iX=C_i$ is equivalent to 
$B_1X^*B_1^{-1}B_iX=C_i$. For $i=1$ this is just $X^*X=I_n$, or, $X^*=X^{-1}$; in 
other words, $X$ is a $*$-unitary element. 
For $i=2, \dots, m$, let $B_i'=B_1^{-1}B_i$ and $C_i'=B_1^{-1}C_i$, and consider 
the system of equations $B_i'X=XC_i'$. If a $*$-unitary element $X$ satisfies  
$B_i'X=XC_i'$ for $i\in\{2, \dots, m\}$, then it also satisfies 
$B_i'^*X=XC_i'^*$ for $i\in\{2, \dots, m\}$. We use the algorithm for module 
isomorphism problem to compute a solution $A\in\GL(n, \F_q)$ to the $2m-2$ 
equations 
$B_i'X=XC_i'$ and $B_i'^*X=XC_i'^*$, $i=2, \dots, m$. Let $D=\{Y\in M(n, \F_q) 
\mid \forall i\in\{2, \dots, m\}, B_i'Y=YB_i', 
B_i'^*Y=YB_i'^*\}$. It is easy to verify that $D$ is a $*$-subalgebra of $M(n, 
\F_q)$. The set of all solutions to 
these equations is just $\{YA : Y\in D\}$. 
The question then becomes to find some such $Y$ so that $(YA)^*YA=I_n$, that is, 
$Y^*Y=(AA^*)^{-1}$. Note that $A^*A$ is a $*$-symmetric element of $D$. The 
problem then becomes to solve the decomposition problem for this $*$-symmetric 
element in the $*$-symmetric algebra $D$, which can be solved using the method from
Section~\ref{sec:iso}. 

The above procedure just differs from the main algorithm in \cite{BFP15} in that 
the latter algorithm does not solve the decomposition problem 
in $D$, but in the smaller (commutative) algebra generated by $A^*A$. 

Now we explain why using extension fields is necessary in \cite{BFP15}. 
Consider the following instance: $m=2$, $B_1=C_1=I_n$, and $B_2=C_2=\diag(\omega, 
\dots, \omega, 1, \dots, 1)$ where $\omega\in\F$ is a non-square and appears $2k$ 
times in $B_2$. Then by techniques from 
Section~\ref{subsubsec:VI}, $B_2=A^tA$ for some $A\in M(n, \F)$. 
As $(B_1, B_2)=(C_1, C_2)$ it is trivial that these two tuples are isometric. 
Following the above procedure, we see that $B_2'=C_2'=B_2=C_2$ as 
$B_1=I_n$ and $*$ is just the transposition. Suppose the algorithm for the module 
isomorphism returns to us $A$ as 
the solution. Note that $A$ is a valid solution to $B_2X=XB_2$, as 
$A^tAA=AA^tA\iff A^tA=AA^t\iff B=B^t$. 
Then we need to solve $Y^tY=(AA^t)^{-1}=B_2^{-1}$ where $Y$ is from the 
centralizing algebra of $B_2$, which is possible by taking $Y=A^{-1}$.
On the other hand, if we insist $Y$ to be from the algebra 
generated by $B_2$, this would not possible over $\F$, by noting that there is no 
diagonal square root of $B_2$ over $\F$. Therefore \cite{BFP15} would need to go 
over an extension field to locate a solution.

\paragraph{Acknowledgements.} 
Part of this research was accomplished while
the first author was visiting the Centre for Quantum Technologies,
National University of Singapore. His research was also partially 
supported by the  Hungarian 
National Research, Development and Innovation Office -- NKFIH, Grant
K115288. 
Y. Q.'s research was supported by the Australian Research Council DECRA 
DE150100720.



\bibliographystyle{alpha}
\bibliography{references}
\end{document}